\documentclass[11pt,onecolumn]{IEEEtran}

\usepackage{amsmath}
\usepackage{tikz}
\usepackage[T1]{fontenc}
\usepackage{newtxtext}
\usetikzlibrary{decorations.pathreplacing, matrix,calc}
\usepackage{textcomp}
\usepackage{xcolor}
\def\BibTeX{{\rm B\kern-.05em{\sc i\kern-.025em b}\kern-.08em
T\kern-.1667em\lower.7ex\hbox{E}\kern-.125emX}}
\usepackage{amssymb,amsthm,amsmath,amsfonts}
\usepackage{enumerate, enumitem}
\usepackage{textcomp}
\usepackage{xcolor}
\usepackage{graphicx}
\usepackage{subcaption}
\usepackage{algorithm}
\usepackage{float}
\usepackage{epstopdf, epsfig}
\usepackage{tcolorbox}
\usepackage{setspace}
\usepackage[noend]{algpseudocode}
\usepackage{cite}
\usepackage{soul}

\newcommand{\ER}{Erd\H{o}s-R{\'e}nyi \ }

\newcommand{\mU}{\mbox{${\mathbf U}$}}

\newcommand{\gre}{\varepsilon}

\newcommand{\gz}{\zeta}

\newcommand{\gt}{\tau}

\newcommand{\go}{\omega}

\newcommand{\SNR}{\ensuremath{\hbox{SNR}}}

\newcommand{\bea}{\begin{array}}
\newcommand{\ena}{\end{array}}
\newcommand{\bds}{\begin {itemize}}
\newcommand{\eds}{\end {itemize}}
\newcommand{\bdf}{\begin{definition}}
\newcommand{\blm}{\begin{lemma}}
\newcommand{\edf}{\end{definition}}
\newcommand{\elm}{\end{lemma}}
\newcommand{\bthm}{\begin{theorem}}
\newcommand{\ethm}{\end{theorem}}
\newcommand{\bprp}{\begin{prop}}
\newcommand{\eprp}{\end{prop}}
\newcommand{\bcl}{\begin{claim}}
\newcommand{\ecl}{\end{claim}}
\newcommand{\bcr}{\begin{coro}}
\newcommand{\ecr}{\end{coro}}
\newcommand{\bquest}{\begin{question}}
\newcommand{\equest}{\end{question}}

\newtheorem{definition}{Definition}
\newtheorem{lemma}{Lemma}
\newtheorem{theorem}{Theorem}
\newtheorem{claim}{Claim}

\newtheorem{corollary}{Corollary}

\usepackage{hyperref}
\usepackage[justification=centering]{caption}
\usepackage{thmtools}

\graphicspath{{simulations/}}
\newcommand{\orcid}[1]{\href{https://orcid.org/#1}{ {\includegraphics[scale=0.5]{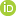}}}}
\begin{document}
\title{Asymptotic Max-Min Fair Allocation with Random Utilities}
\author{Noam Glazner\textsuperscript{\orcid{0009-0003-1524-1157}} and 
  Amir Leshem\thanks{Authors are with Faculty of Engineering, Bar-Ilan University, Ramat Gan 52900, Israel. Email: amir.leshem@biu.ac.il. This research was partially supported by grant ISF 2197/22 (Corresponding author: Amir Leshem}\textsuperscript{ 
  \orcid{0000-0002-2265-7463}}, ~\IEEEmembership{Fellow,~IEEE}
}
\maketitle
\begin{abstract} 
We investigate the asymptotic behavior of max-min fair allocations for indivisible goods under i.i.d. random utilities. For 
$N$ agents and 
$K$ goods with utilities 
${\mU_{i,j}}$
drawn independently from a common distribution 
$F$, we derive asymptotic characterizations of the max-min value in the balanced case 
$K=N$ (and in 
$K=LN$ extensions) via distributional quantiles. We then study the efficiency impact of max-min fairness by comparing the resulting total welfare with the optimal sum welfare. For distributions with sufficiently light tails, we prove that the relative efficiency loss converges to zero as the market grows, implying that max-min fairness incurs negligible welfare loss in large random instances for a broad class of distributions.
\end{abstract}
\section{Introduction}
\label{sec:intro}
Fair resource allocation is a key problem in the field of resource allocation \cite{kelly1998rate, jain1984quantitative, lan2010axiomatic, mo2002fair}. 
Fairness considerations date back to Nash's bargaining analysis of the bargaining problem  \cite{nash1950bargaining}. Kalai and Smorodinsky proposed different fairness axioms \cite{kalai1975other}. Aumann and Maschler proposed a generalization of the Talmudic division rule \cite{aumann1985game}. These solutions can be used to provide fair allocation of divisible goods \cite{han2005fair, leshem2006bargaining, leshem2008cooperative} and provide approximations to the indivisible case. An overview of these appeared in \cite{leshem2009game}. 
A surge of research has focused specifically on the indivisible case, exploring fairness notions such as envy-freeness and $\alpha$-fairness \cite{akrami2025achieving, bilo2026approximately, montanari2025weighted, igarashi2019pareto, golz2026fair, dror2023fair, jin2018trade, massoulie2007structural}.
The $\alpha$-fairness family provides a continuous interpolation between efficiency-oriented and fairness-oriented allocations: the case $\alpha=0$ reduces to the utilitarian objective of maximizing the sum-utility. Proportional fairness appears as the special case $\alpha=1$, while max-min fairness arises as the limiting regime of $\alpha\to\infty$, and is fundamentally different \cite{bistritz2020my}. 

Max-min fairness is especially prominent due to its rigorous guarantees for the worst-off users. This objective has wide-ranging applications, including in economics \cite{mongin2021rawls, stark2014reconciling}, artificial intelligence \cite{park2026multi, park2024max, bistritz2021one, leshem2025near, radunovic2007unified},
 and in various areas of networking and communication systems \cite{zehavi2013weighted, trassl2022outage, zheng2018joint, yu2025max, nguyen2020max }.
\newline 
The problem of allocating indivisible resources among agents according to a max-min objective is also known as egalitarian social welfare maximization. When the number of agents equals the number of resources, the problem reduces to the linear bottleneck assignment problem \cite{burkard2012assignment}, which can be solved in polynomial time in the number of resources. In contrast, when the number of resources exceeds the number of agents, the problem is sometimes referred to as the Santa Claus problem \cite{bansal2006santa}, and it is shown to be NP-hard to approximate within a factor of better than $\frac{1}{2}$ \cite{golovin2005max}. Therefore, when the number of resources is larger than the number of agents, many studies focus on improving the computational complexity and accuracy of these approximations.
A restricted case approximation algorithm for the max-min fairness value was proposed in \cite{bansal2006santa}. The first polynomial-time approximation algorithm for the general problem was provided in \cite{asadpour2007approximation}, which was improved in \cite{haeupler2011new}.
 An approximation algorithm that approximates the solution up to a single resource for each agent using a relaxation of this problem that can be solved in polynomial time was provided in \cite{zehavi2013weighted}. A more comprehensive discussion of the problem can be found in \cite{lang2024fair}.
These approximations have proven effective across a wide range of applications. When utilities are stochastic and modeled as random variables rather than fixed numerical values, studying the behavior of the max-min value becomes particularly interesting. In particular, the relationship between utility distributions and max-min outcomes raises several important theoretical questions.

Quantifying the max-min value is essential to understanding the trade-off between fairness and efficiency. The choice of objective function plays a pivotal role.
When some agents receive utilities that are significantly higher than those of others, a max-min allocation tends to favor the "poorer" agents by allocating more resources to them, and vice versa. In contrast, the optimal sum welfare allocation (i.e., maximizing the sum of the total utility of all the agents) naturally allocates more resources to the "richer" agents, leading to a potential conflict between fairness and efficiency. ~\cite{bertsimas2011price} studied this loss of efficiency and defined the price of fairness as
\begin{equation}
\mathrm{POF}(\mU)
=
\frac{\mathrm{S}(\mU)-\mathrm{FAIR}(\mU)}
     {\mathrm{S}(\mU)}.
\end{equation}
where $\mathrm{S}(\mU)$ is the optimal sum welfare value, and
$\mathrm{FAIR}(\mU)$ is the total utility achieved by the fair
allocation, which in our setting is the max-min fair allocation.
Their worst-case bounds show that this efficiency loss may approach
one as the number of agents increases.

The price of fairness has since been studied under various fairness
criteria. \cite{caragiannis2012efficiency} analyzed
proportionality, envy-freeness, and equitability. 
\cite{barman2020optimal} established tight
$\Theta(\sqrt{N})$ bounds for envy-freeness up to one good and
one-half maximin-share guarantees.  Under our definition, this
corresponds to
$\mathrm{POF}=1-\Theta\left(\frac{1}{\sqrt{N}}\right)$. 

The trade-off between fairness and efficiency is particularly
interesting when utilities are independent and identically distributed
(i.i.d.). In this setting, the similarity among the agents' utilities
may lead to a small Price of Fairness. In the balanced regime \(N=K\),
the max-min allocation problem is closely related to the random linear
bottleneck assignment problem. ~\cite{pferschy1996random}
studied its asymptotic expected optimal value and derived explicit
bounds for uniformly distributed costs. ~\cite{spivey2011asymptotic}
subsequently obtained refined asymptotic expansions for moments of the
optimal bottleneck value under additional distributional conditions.
In contrast, our results are formulated directly in terms of upper-tail
behavior and provide quantile-based asymptotic characterizations for
broad, nested classes of utility distributions.

When $K>N$,
\cite{bistritz2018asymptotically} constructed, under
more restrictive assumptions, a max-min fair allocation whose total
utility is asymptotically close to the optimal sum welfare value. While their analysis is restricted to a specific setting, we provide a comprehensive asymptotic characterization of the problem, establishing stronger theoretical results for a broad class of utility distributions and problem configurations. Our central question is therefore:
under the i.i.d.\ assumption, does max-min fairness cause a
significant loss in total welfare?

\textbf{Novelty and contributions:}
The central contribution of this work is a general tail-based asymptotic theory for max-min fair allocation under i.i.d.\ random utilities. Rather than analyzing particular utility distributions separately, we identify explicit conditions on the upper-tail quantiles of the common utility distribution that determine the asymptotic behavior of the max-min value. In the balanced regime $K=N$, the max-min problem is linked exactly to the perfect-matching threshold of a random bipartite graph. This yields the distribution-free first-order characterization
\begin{equation}
M_{N,N}\sim F^{-1}\!\left(1-\frac{\ln N}{N}\right)
\end{equation}
in probability for the broad class of distributions with regularly varying upper-tail quantiles, as well as for distributions with a finite positive upper endpoint. For the more restrictive admissible Weibull-type class, we strengthen this relative characterization to additive convergence. Thus, the relevant distinction between utility distributions is expressed directly through their upper-tail behavior, and progressively stronger tail regularity yields progressively stronger forms of convergence.

This viewpoint is substantially different from the distribution-specific asymptotic calculations previously available for the random bottleneck assignment problem. In particular, the refined moment expansions in~\cite{spivey2011asymptotic} require detailed analysis of individual distribution families. Our results instead isolate general upper-tail conditions under which the same asymptotic law holds simultaneously for broad classes of distributions, including light-tailed, lognormal, and regularly varying heavy-tailed models. The resulting hierarchy separates the assumptions needed for relative convergence from those sufficient for additive convergence and makes explicit the role played by the heaviness and local regularity of the upper tail.

We further show that this tail-based characterization extends well beyond the balanced assignment problem. In the proportional-growth regime $K=LN$, with fixed integer $L$ and $N\to\infty$, we prove
\begin{equation}
M_{N,LN}\sim
L F^{-1}\!\left(1-\frac{\ln N}{N}\right)
\end{equation}
in probability for Weibull-type utilities and distributions with a finite positive upper endpoint. The proof combines a new $L$-matching construction with an asymptotically sharp sum-welfare upper bound. This gives a direct characterization of the unrestricted max-min optimum, rather than of a particular allocation algorithm, and covers broad distribution classes that include the fading models considered in~\cite{bistritz2018asymptotically}. In the complementary fixed-agent regime, where $N$ is fixed and $K\rightarrow\infty$, we prove that the opportunistic allocation that assigns each resource to its highest-utility agent is asymptotically max-min optimal, together with an explicit finite-$K$ probabilistic bound.

Finally, these asymptotic characterizations yield a strong efficiency consequence. For the proportional-growth regime under the stated light-tail assumptions, and for the fixed-agent regime under finite-moment assumptions, the total welfare of \emph{every} max-min allocation is asymptotically equivalent to the optimal sum welfare. Consequently,
\begin{equation}
\operatorname{PoF}\xrightarrow{\mathrm P}0.
\end{equation}
Thus, although max-min fairness may incur a severe efficiency loss in worst-case indivisible-allocation instances, under broad i.i.d.\ random utility models this loss disappears asymptotically. The results therefore identify a large stochastic regime in which strong fairness and asymptotic efficiency are simultaneously achievable.

The paper is organized as follows. Section~\ref{sec:problem formulation} formulates the allocation problem and provides definitions of the basic distribution families.

We begin with the important case of equal number of agents and resources, in Section~\ref{sec: Single Resource Per agent}. We derive a complete asymptotic characterization of the max-min value in terms of the upper tail of the utility distribution using a random-graph-based matching-theoretic approach. The section ends with explicit computation of the max-min asymptotics for Rayleigh fading channels. 

In Section~\ref{sec:multiple_resources}, we extend our analysis to the case where the number of resources exceeds the number of agents ($K > N$). We outline two distinct asymptotic models: the proportional growth model and the fixed-agent model.
Section~\ref{sec:first_model} examines the proportional growth model ($K=LN$), where both the number of agents and resources grow to infinity while maintaining a constant ratio. We establish a lower bound on the
max-min value by extending the random bipartite-graph approach  to a more general matching construction that accommodates
multiple resources per agent.
Section~\ref{sec:second_model} addresses the fixed-agent model ($K \gg N$), where the number of agents is fixed while the number of resources grows asymptotically. We demonstrate that in this regime, the greedy allocation policy is asymptotically optimal for the max-min objective.

Section~\ref{sec:simulations} presents numerical experiments validating our theoretical bounds and asymptotic results across all discussed regimes.
Finally, section~\ref{sec:conclusion} summarizes our main findings, discusses the limitations of the analysis, and outlines directions for future research.
\section{Problem Formulation}
\label{sec:problem formulation}
Throughout the paper, uppercase symbols denote random objects,
lowercase symbols denote their realizations, and bold symbols denote
matrices.

Assume that we have a set
\begin{equation}
A=\{1,\ldots,N\}
\end{equation}
of \(N\) agents and a set
\begin{equation}
B=\{1,\ldots,K\}
\end{equation}
of \(K\) resources, where \(K\geq N\). Let \(U_{i,j}\) denote
the random utility of resource \(j\) to agent \(i\). We assume that
the random variables \(U_{i,j}\) are nonnegative and i.i.d. with
common continuous distribution function \(F\).

For every pair \((N,K)\), let
\begin{align}
\mU_{(N,K)}
&=
\left(U_{i,j}\right)_
{\substack{1\leq i\leq N\\1\leq j\leq K}}
\in\mathbb{R}_+^{N\times K}
\end{align}
denote the random utility matrix associated with \(N\) agents and
\(K\) resources. A realization of \(\mU_{(N,K)}\) is denoted by
\begin{equation}
\mathbf{u}_{(N,K)}
=
\left(u_{i,j}\right)_
{\substack{1\leq i\leq N\\1\leq j\leq K}}.
\end{equation}
When the dimensions are clear from the context, we simply write
\(\mathbf u\).

Let
\begin{equation}
\mathcal{G}_{N,K}
:=
\{g:B\to A\}
\end{equation}
denote the set of feasible allocations. For a realization
\(\mathbf u\), an allocation \(g\in\mathcal{G}_{N,K}\), and an agent
\(n\in A\), define the total utility received by agent \(n\) as
\begin{align}
t_n(g;\mathbf u)
&=
\sum_{k:g(k)=n}u_{n,k}.
\label{def:maxmin_perm_K=LN}
\end{align}

For a fixed realization \(\mathbf u\), define the max-min value by
\begin{align}
m_{N,K}(\mathbf u)
&:=
\max_{g\in\mathcal{G}_{N,K}}
\min_{1\leq n\leq N}t_n(g;\mathbf u),
\label{def:deterministic-max-min}
\end{align}
and let
\begin{align}
\mathcal{G}_{\mathrm{mm}}(\mathbf u)
&:=
\arg\max_{g\in\mathcal{G}_{N,K}}
\min_{1\leq n\leq N}t_n(g;\mathbf u)
\end{align}
denote the set of max-min allocations. We write
\begin{equation}
g_{\mathrm{mm}}(\mathbf u)
\in
\mathcal{G}_{\mathrm{mm}}(\mathbf u)
\end{equation}
for any selected max-min allocation.

The corresponding random max-min value is
\begin{align}
M_{N,K}
&:=
m_{N,K}\left(\mU_{(N,K)}\right).
\label{def: max-min value multiple resources}
\end{align}

Because the utilities are nonnegative and \(F\) is continuous,
\begin{equation}
\Pr(U_{i,j}=0)=0.
\end{equation}
Therefore, for every fixed \(N\) and \(K\), all entries of
\(\mU_{(N,K)}\) are strictly positive almost surely.

For every strictly positive realization \(\mathbf u\), every
allocation in \(\mathcal{G}_{\mathrm{mm}}(\mathbf u)\) is surjective.
Indeed, because \(K\geq N\), there exists a surjective allocation,
and every surjective allocation gives each agent strictly positive
total utility. Hence,
\begin{equation}
m_{N,K}(\mathbf u)>0.
\end{equation}
In contrast, a non-surjective allocation leaves at least one agent
with utility zero and therefore cannot attain \(m_{N,K}(\mathbf u)\).
Consequently, every max-min allocation of
\(\mU_{(N,K)}\) is surjective almost surely.

For a fixed realization \(\mathbf u\), define the optimal sum-welfare
value \cite{bistritz2018asymptotically} by
\begin{align}
s_{N,K}(\mathbf u)
&:=
\max_{g\in\mathcal{G}_{N,K}}
\sum_{n=1}^{N}t_n(g;\mathbf u)
\nonumber\\
&=
\sum_{k=1}^{K}\max_{1\leq n\leq N}u_{n,k}.
\label{def:deterministic-sum-welfare}
\end{align}
The corresponding random optimal sum-welfare value is
\begin{align}
S_{N,K}
&:=
s_{N,K}\left(\mU_{(N,K)}\right).
\label{def: sum-rate value multiple resources}
\end{align}

For every realization \(\mathbf u\) and every
\(g_{\mathrm{mm}}(\mathbf u)\in
\mathcal{G}_{\mathrm{mm}}(\mathbf u)\),
\begin{equation}
N m_{N,K}(\mathbf u)
\leq
\sum_{n=1}^{N}
t_n\bigl(g_{\mathrm{mm}}(\mathbf u);\mathbf u\bigr)
\leq
s_{N,K}(\mathbf u).
\end{equation}
Applying this deterministic inequality to the random matrix
\(\mU_{(N,K)}\) gives
\begin{equation}
\label{eq:general_upper_bound}
M_{N,K}
\leq
\frac{1}{N}S_{N,K}
\qquad\text{almost surely}.
\end{equation}

In this work, we examine the following asymptotic regimes:
\begin{enumerate}
    \item \(K=N\) and \(N\to\infty\);
    \item \(K=LN\), where \(L\geq1\) is a fixed integer, and
          \(N\to\infty\);
    \item \(N\) is fixed and \(K\to\infty\).
\end{enumerate}
Before discussing the solution, we provide an overview of the important families of distributions used in this paper. 
\subsection{Distribution Classes and Tail Hierarchy}
\label{sec:distribution-classes}

The distribution of the random max-min value
\begin{equation}
M_{N,K}
=
m_{N,K}\left(\mU_{(N,K)}\right)
\end{equation}
is governed by the common distribution \(F\) of the random utility
entries \(U_{i,j}\), and in particular by its upper-tail behavior. Furthermore, we translate
the max-min problem into a graph matching problem.
Consequently, rather than analyzing various distributions
separately, we formulate the results directly in terms of tail
growth. This provides a unified characterization covering
a very broad family of utility distributions and makes
explicit how progressively lighter tails yield progressively
stronger forms of convergence.
We now provide a short overview of the basic definitions
and properties of the distributions relevant to this paper. We define the three unbounded-support distribution
classes used throughout the paper. They are presented
from the broadest to the most restrictive and satisfy the
strict hierarchy
\begin{equation}
\label{eq:distribution-class-hierarchy}
\mathcal{C}_{\mathrm{RVQ}}
\supsetneq
\mathcal{C}_{\mathrm{W}}
\supsetneq
\mathcal{C}_{\mathrm{AW}},
\end{equation}
where the subscripts denote, respectively, regularly varying
upper-tail quantiles, Weibull-type tails, and admissible Weibull-type
tails. For the standard distribution families considered in this
paper, moving from left to right corresponds to restricting attention
to progressively lighter and more regular upper tails.

The following definitions are stated for a generic continuous
distribution function \(F\) supported on \([0,\infty)\) and having
an unbounded upper endpoint. In the allocation model, Let
\begin{equation}
F^{-1}(u)
:=
\inf\{x\in\mathbb{R}:F(x)\geq u\},
\qquad u\in(0,1),
\end{equation}
denote the generalized inverse of \(F\). Define the cumulative hazard
function, the associated log-tail quantile, and the upper-tail
quantile by
\begin{equation}
\label{eq:log-tail-quantile-definition}
\begin{aligned}
H_F(x)
&:=
-\ln\bigl(1-F(x)\bigr),\\
V_F(y)
&:=
H_F^{\leftarrow}(y)
=
F^{-1}\left(1-e^{-y}\right),
&& y>0,\\
Q_F(t)
&:=
F^{-1}\left(1-\frac{1}{t}\right)
=
V_F(\ln t),
&& t>1.
\end{aligned}
\end{equation}
Thus, \(V_F\) and \(Q_F\) describe the same upper quantiles on two
different scales: the exceedance probability is \(e^{-y}\) for
\(V_F(y)\) and \(1/t\) for \(Q_F(t)\).

Recall that a positive function \(\ell\) is slowly varying at infinity
if
\begin{equation}
\frac{\ell(ct)}{\ell(t)}
\longrightarrow 1
\qquad\text{as }t\to\infty
\end{equation}
for every fixed \(c>0\).

\begin{definition}[Regularly varying upper-tail quantile]
\label{def:regularly-varying-upper-tail-quantile}
We say that \(F\) has a \emph{regularly varying upper-tail quantile}
with coefficient \(\rho\in[0,\infty)\) if \(Q_F\) is regularly
varying at infinity with index \(\rho\); that is, for every \(c>0\),
\begin{equation}
\label{eq:regularly-varying-upper-tail-quantile}
\frac{Q_F(ct)}{Q_F(t)}
\longrightarrow c^\rho
\qquad\text{as }t\to\infty.
\end{equation}
Equivalently,
\begin{equation}
Q_F(t)=t^\rho\ell(t)
\end{equation}
for some slowly varying function \(\ell\).
\end{definition}

This is a deliberately general class. It contains many standard
positive unbounded distributions used in applications, including all
Weibull-type distributions defined below, the lognormal distribution,
and heavy-tailed distributions such as the Pareto type II, and Burr
distributions. It therefore permits relative asymptotic
characterizations without excluding practically important power-law
models.

The case \(\rho=0\) consists of slowly varying upper-tail quantiles,
whereas \(\rho>0\) permits polynomial quantile growth and regularly
varying heavy tails. In particular, a Pareto type II distribution belongs to
\(\mathcal{C}_{\mathrm{RVQ}}\) with a positive coefficient but does
not belong to \(\mathcal{C}_{\mathrm{W}}\). This shows that the first
inclusion in \eqref{eq:distribution-class-hierarchy} is strict.

\begin{definition}[Weibull-type tail]
\label{def:Weibull-type-tail}
We say that \(F\) has a \emph{Weibull-type tail} with coefficient
\(\theta\in[0,\infty)\) if \(V_F\) is regularly varying at infinity
with index \(\theta\); that is, for every \(c>0\),
\begin{equation}
\label{eq:Weibull-tail-regular-variation}
\frac{V_F(cy)}{V_F(y)}
\longrightarrow c^\theta
\qquad\text{as }y\to\infty.
\end{equation}
Equivalently,
\begin{equation}
V_F(y)=y^\theta L(y),
\end{equation}
where \(L\) is slowly varying at infinity.
\end{definition}

The Weibull-type class restricts attention to tails that are lighter
than the power-law tails permitted by
\(\mathcal{C}_{\mathrm{RVQ}}\). Within standard parametric families,
the coefficient \(\theta\) describes the growth rate of extreme upper
quantiles: larger values of \(\theta\) correspond to faster-growing
quantiles and hence to heavier tails. The class is well studied; see,
for example, \cite{broniatowski1993estimation}.

A Weibull distribution with shape parameter \(\beta>0\) has
coefficient \(\theta=1/\beta\). Exponential, gamma, and chi-squared
distributions have coefficient \(\theta=1\), whereas the half-normal,
Rayleigh, chi, Maxwell, Nakagami, and Rice distributions have
coefficient \(\theta=1/2\). The positive Gompertz distribution is an
example with coefficient \(\theta=0\).

No differentiability of \(V_F\) is required in this definition, and
every finite coefficient \(\theta\geq0\) is permitted. Moreover, every
Weibull-type distribution has a slowly varying upper-tail quantile
\(Q_F\) and therefore belongs to
\(\mathcal{C}_{\mathrm{RVQ}}\) with coefficient \(\rho=0\).

\begin{definition}[Admissible Weibull-type tail]
\label{def: Admissible Weibull-type tail}
We say that \(F\) has an \emph{admissible Weibull-type tail} if, for
some \(\theta\in[0,2)\),
\begin{equation}
V_F(y)=y^\theta L(y),
\end{equation}
where \(L\) is slowly varying, \(V_F\) is eventually continuously
differentiable, and
\begin{equation}
\label{eq:admissible-Weibull-derivative}
\frac{yV_F'(y)}{V_F(y)}
\longrightarrow\theta
\qquad\text{as }y\to\infty.
\end{equation}
\end{definition}

The admissible class further restricts the Weibull-type class to
sufficiently light and regular tails. The derivative condition
controls the local variation of nearby extreme quantiles, while the
restriction \(\theta<2\) ensures that the relevant quantile gaps
vanish not only relatively but also in absolute value. These
conditions permit the additive convergence results established later
in the paper.

The admissible class includes the smooth examples above whenever
their coefficient is below two. In particular, it includes the
exponential, gamma, chi-squared, half-normal, Rayleigh, chi, Maxwell,
Nakagami, Rice, and Gompertz distributions. A Weibull distribution is
admissible when its shape parameter satisfies \(\beta>1/2\). When
\(\beta\leq1/2\), it remains Weibull-type but is not admissible because
\begin{equation}
\theta=\frac{1}{\beta}\geq2.
\end{equation}
This shows that the second inclusion in
\eqref{eq:distribution-class-hierarchy} is strict.

The containment relations and the quantile comparisons required for
the matching-threshold analysis are proved in
Lemma~\ref{lemma:regular-tail-quantile-threshold} in
Appendix~\ref{app:regular-tail-quantiles}.
\section{A Single Resource per Agent}
\label{sec: Single Resource Per agent}
This section focuses on the balanced allocation setting ($N=K$). 
In that case, let $U$ be a random variable with continuous distribution function $F$, and,
for each $N$, let $\mU_{(N,N)}$ be an $N\times N$ random utility
matrix whose entries are i.i.d. copies of $U$.

\begin{theorem}
\label{theorem:main}
Suppose that the common utility distribution \(F\) has a regularly
varying upper-tail quantile or a finite positive upper endpoint. Then,
for every \(\gre>0\),
\begin{align}
\lim_{N\to\infty}
\Pr\left(
\left|
1-
\frac{M_{N,N}}
{F^{-1}\left(1-\frac{\ln N}{N}\right)}
\right|
>\gre
\right)
=0.
\tag{a}
\end{align}

If, more restrictively, \(F\) has an admissible Weibull-type tail or
a finite positive upper endpoint, then, for every \(\gre>0\),
\begin{align}
\lim_{N\to\infty}
\Pr\left(
\left|
M_{N,N}
-
F^{-1}\left(1-\frac{\ln N}{N}\right)
\right|
>\gre
\right)
=0.
\tag{b}
\end{align}
\end{theorem}

Thus, part~(a) gives relative convergence in probability, whereas
part~(b) gives the stronger additive convergence in probability. 

To prove Theorem~\ref{theorem:main}, we first translate the
deterministic max-min problem into a perfect-matching problem. We then apply this representation to
the random matrix \(\mU_{(N,N)}\) and use the perfect-matching threshold
for \ER random bipartite graphs.

\begin{claim}
\label{claim:K=N}
Let
\begin{equation}
\mathbf u
=
(u_{i,j})_{1\leq i,j\leq N}
\in\mathbb{R}_+^{N\times N}
\end{equation}
be a fixed realization of the utility matrix, and let
\(\gt\in\mathbb{R}\). Define the bipartite threshold graph
\begin{equation}
G_{\mathbf u}(\gt)
=
\bigl(A,B,E_{\mathbf u}(\gt)\bigr),
\end{equation}
where \(A=B=\{1,\ldots,N\}\) and
\begin{equation}
E_{\mathbf u}(\gt)
=
\left\{
(i,j)\in A\times B:u_{i,j}\geq\gt
\right\}.
\end{equation}
Then
\begin{equation}
\gt\leq m_{N,N}(\mathbf u)
\end{equation}
if and only if \(G_{\mathbf u}(\gt)\) contains a perfect matching.
\end{claim}
\begin{proof}
If \(G_{\mathbf u}(\gt)\) contains a perfect matching, the matching
defines a bijective allocation \(g\) satisfying
\begin{equation}
u_{g(j),j}\geq\gt
\qquad\text{for every }j\in\{1,\ldots,N\}.
\end{equation}
Thus, every agent receives utility at least \(\gt\), and hence
\begin{equation}
m_{N,N}(\mathbf u)\geq\gt.
\end{equation}

Conversely, suppose that \(m_{N,N}(\mathbf u)\geq\gt\). If
\(\gt\leq0\), then \(G_{\mathbf u}(\gt)\) is complete because
\(\mathbf u\) is nonnegative, so it contains a perfect matching.

Now suppose that \(\gt>0\), and let
\(g^\star\in\mathcal G_{\mathrm{mm}}(\mathbf u)\). Since
\begin{equation}
t_i(g^\star;\mathbf u)\geq m_{N,N}(\mathbf u)\geq\gt>0
\qquad\text{for every }i,
\end{equation}
every agent receives at least one resource. Because there are \(N\)
agents and \(N\) resources, \(g^\star\) is bijective and every agent receives exactly one resource. Consequently,
\begin{equation}
u_{g^\star(j),j}
=
t_{g^\star(j)}(g^\star;\mathbf u)
\geq\gt
\qquad\text{for every }j,
\end{equation}
so the edges induced by \(g^\star\) form a perfect matching in
\(G_{\mathbf u}(\gt)\).
\end{proof}

We now apply Claim~\ref{claim:K=N} to the random matrix
\(\mU_{(N,N)}\). Let
\begin{equation}
G_{\mU}(\gt):= G_{\mU_{(N,N)}}(\gt)
\end{equation}
denote the random threshold graph obtained by replacing \(u_{i,j}\)
with \(U_{i,j}\) in the definition of \(G_{\mathbf u}(\gt)\). The claim
gives the event identity
\begin{equation}
\label{eq:matching-event-identity}
\left\{M_{N,N}\geq\gt\right\}
=
\left\{
G_{\mU}(\gt)
\text{ contains a perfect matching}
\right\}.
\end{equation}

Because the entries \(U_{i,j}\) are i.i.d. and \(F\) is continuous,
the edges of \(G_{\mU}(\gt)\) are independent and have common
probability
\begin{equation}
p_{\gt}
=
\Pr(U_{i,j}\geq\gt)
=
1-F(\gt).
\end{equation}
Thus, \(G_{\mU}(\gt)\) is an \ER random bipartite
graph.

We use the following perfect-matching threshold theorem of \ER \cite{erdos1966random,erdos1968random}; see also
\cite[Theorem~6.1]{frieze2016introduction}.

\begin{theorem}
\label{theorem:ER-perfect-matching}
Let \(\go_N\) be a real sequence and let
\begin{equation}
p_N=\frac{\ln N+\go_N}{N},
\end{equation}
where \(p_N\in[0,1]\) for all sufficiently large \(N\). Let
\(E_N\) denote the event that the random bipartite graph
\(G_{N,N,p_N}\) contains a perfect matching. Then
\begin{equation}
\lim_{N\to\infty}\Pr(E_N)
=
\begin{cases}
0,
& \go_N\to-\infty,\\
e^{-2e^{-c}},
& \go_N\to c\in\mathbb{R},\\
1,
& \go_N\to+\infty.
\end{cases}
\end{equation}
\end{theorem}

Using the theorem we obtain that for every $\gre$ choosing \begin{align} F(\gt_N)=1-\frac{\ln N+\ln \ln \ln N}{N} \end{align} the graph $G_{\mathbf u}(\gt_N)$ is an \ER random graph with edge probability $\frac{\ln N + \ln \ln \ln N}{N}$. Hence it has a perfect matching with probability approaching $1$. Similarly, choosing $F(\gt_N)=1-\frac{\ln N-\ln \ln \ln N}{N}$, the graph $G_{\mathbf u}(\gt_N)$ has no perfect matching with probability approaching 1. Hence the max-min value $M_{N,N}$ satisfies with high probability: \begin{align} F^{-1}\left(1-\gz_{1}^N\right )\le M_{N,N} \le F^{-1}\left(1-\gz_2^N\right). \label{eq: K=N lower bound} \end{align} where $\gz_1^N=\frac{\ln N + \ln \ln \ln N}{N}, \gz_2^N=\frac{\ln N - \ln \ln \ln N}{N}$. To complete the proof, we need to show that $F^{-1}_U\left(1-\gz_{1}^N\right )$ is asymptotically close to $F^{-1}\left(1-\gz_2^N\right)$ or, under the weaker assumptions, asymptotically equivalent in ratio. This is true for Admissible Weibull-type tail class and for regularly varying upper-tail class, (including bounded distributions), accordingly, as we show in the following lemmas: \begin{restatable}{lemma}{admissible} \label{Lemma: total convergence order statistic complete} Let $F$ be an Admissible Weibull-type tail as defined in Definition \ref{def: Admissible Weibull-type tail} or a finite positive upper endpoint, then: \begin{equation} F^{-1}\left( 1-\gz_1^N \right) - F^{-1}\left( 1-\gz_2^N \right) \longrightarrow0. \end{equation} \end{restatable} \begin{proof} See appendix \ref{app: special case}. \end{proof} The max-min value converges in probability to $F^{-1}\left(1-\frac{\ln N}{N}\right)$ and this proves (b). \begin{restatable}{lemma}{regular} \label{lemma:regular-tail-quantile-threshold} Let $F$ be a continuous distribution function. Then: If $F$ has a regularly varying upper-tail quantile or a finite positive upper endpoint, then, with \begin{equation} \gz_1^N:=\frac{\ln N+\ln\ln\ln N}{N}, \qquad \gz_2^N:=\frac{\ln N-\ln\ln\ln N}{N}, \end{equation} we have \begin{equation} \label{eq:regular-tail-quantile-matching-ratio} \frac{F^{-1}(1-\gz_1^N)}{F^{-1}(1-\gz_2^N)} \longrightarrow1. \end{equation} \end{restatable} \begin{proof} See appendix~\ref{app:regular-tail-quantiles}. \end{proof} This proves (a).

We now show a closed form expression for the max-min asymptotic behaviour for the important case of parallel, i.i.d Rayleigh fading channels. In this case the signal to noise ratio is exponentially distributed. In this case 
\begin{align}
    F(t)=\left\{
    \begin{tabular}{cl}
        $0$ & if $t<0$  \\
        $1-e^{-t}$ & if $t\ge0$.
    \end{tabular}
    \right.
\end{align}
Hence, as the number of channels grows, the max-min fair allocation provides signal-to-noise ratio satisfying:
\begin{align}
    F(\SNR)=1-e^{-\SNR}=1-\frac{\ln N}{N}.
\end{align}
Hence, 
\begin{align}
    SNR_{\max \min} - (\ln N- \ln \ln N) \xrightarrow[]{p} 0.
\end{align}

\section{Multiple Resources per Agent}
\label{sec:multiple_resources}

We now consider the max-min allocation problem when the number of
resources exceeds the number of agents (\(K>N\)). As previously noted,
this problem is NP-hard. Consequently, computing the exact max-min
value is computationally intractable in general, making theoretical
bounds useful for benchmarking suboptimal heuristics. Nevertheless, we
show that, similarly to the balanced case (\(N=K\)), the random
max-min value \(M_{N,K}\) can be asymptotically bounded using only the
common utility distribution \(F\).

We distinguish between two asymptotic regimes. The first, the
proportional-growth model, assumes proportional scaling of the numbers
of agents and resources. Specifically, we let \(K=LN\), where \(L\)
is a fixed positive integer, and examine the asymptotic behavior as
\(N\to\infty\). The second regime is motivated by settings such as
those considered in \cite{zehavi2013weighted}, in which a relatively
small number of agents share a large pool of resources. Accordingly,
the fixed-agent model assumes that \(N\) is fixed and examines the
asymptotic behavior as \(K\to\infty\).

As shown in \eqref{eq:general_upper_bound}, the normalized optimal
sum-welfare value \(S_{N,K}/N\) is an upper bound for \(M_{N,K}\). The
following claim describes the distribution of \(S_{N,K}\).

\begin{claim}
\label{Remark: sum-rate}
Let
\begin{equation}
\mU_{(N,K)}
=
\left(U_{i,j}\right)_
{\substack{1\leq i\leq N\\1\leq j\leq K}}
\end{equation}
be the random utility matrix defined in the problem formulation. Then
\begin{equation}
S_{N,K}
=
s_{N,K}\left(\mU_{(N,K)}\right)
=
\sum_{j=1}^{K}\max_{1\leq i\leq N}U_{i,j}.
\end{equation}
Moreover, if
\begin{equation}
U_{[N:N]}^{(1)},\ldots,U_{[N:N]}^{(K)}
\end{equation}
are independent copies of the maximum of \(N\) i.i.d. random
variables with distribution function \(F\), then
\begin{equation}
S_{N,K}
\stackrel{d}{=}
\sum_{j=1}^{K}U_{[N:N]}^{(j)}.
\end{equation}
\end{claim}

In the following sections, we use this representation of the upper
bound to analyze the asymptotic behavior of \(M_{N,K}\). For the
fixed-agent model, we also provide a probabilistic lower bound for
finite \(K\) and derive an upper bound on the corresponding error
probability. Due to their simplicity, these bounds provide a practical
tool for evaluating the efficiency of allocation algorithms.
\section{Proportional Growth Model: Numerous Agents, More Resources}
\label{sec:first_model}

In this section, we consider the proportional-growth regime in which
\(K=LN\), where \(L\) is a fixed positive integer. Thus, as
\(N\to\infty\), both the number of agents and the number of resources
grow, while
\begin{equation}
\frac{K}{N}=L
\end{equation}
remains fixed. Accordingly, we denote the random utility matrix and
its associated random max-min and optimal sum-welfare values by
\begin{equation}
\mU_{(N,LN)},\qquad M_{N,LN},\qquad S_{N,LN},
\end{equation}
respectively.

The following theorem characterizes the asymptotic max-min value in the
proportional-growth regime \(K=LN\). It applies to distributions with
Weibull-type tails as well as to distributions with a finite upper
endpoint.

\begin{theorem}
\label{THM:main_theorem_order_statistic}
Let \(L\) be a fixed positive integer, and let
\begin{equation}
\mU_{(N,LN)}
=
(U_{i,j})_{\substack{1\leq i\leq N\\1\leq j\leq LN}}
\end{equation}
be a random utility matrix whose entries are i.i.d. with common
continuous distribution function \(F\). Let
\begin{equation}
M_{N,LN}
=
m_{N,LN}\left(\mU_{(N,LN)}\right)
\end{equation}
be the corresponding random max-min value. If \(F\) has a
Weibull-type tail or a finite positive upper endpoint, then, for every
\(\gre>0\),
\begin{align}
\lim_{N\to\infty}
\mathbb{P}\left(
\left|
1-
\frac{M_{N,LN}}
{L F^{-1}\left(1-\frac{\ln N}{N}\right)}
\right|>\gre
\right)
=0.
\end{align}
\end{theorem}

Before presenting the proof of
Theorem~\ref{THM:main_theorem_order_statistic}, we record an important
consequence concerning the price of fairness. The theorem implies that
the welfare loss caused by imposing max-min fairness is asymptotically
negligible.

\begin{corollary}
\label{cor:vanishing-pof}
Under the assumptions of
Theorem~\ref{THM:main_theorem_order_statistic}, let
\begin{equation}
g_{\mathrm{mm}}
\in
\mathcal G_{\mathrm{mm}}\left(\mU_{(N,LN)}\right)
\end{equation}
be any max-min allocation of the random utility matrix, and define
\begin{equation}
\operatorname{FAIR}_{N,LN}
:=
\sum_{i=1}^{N}
t_i\left(g_{\mathrm{mm}};\mU_{(N,LN)}\right)
\end{equation}
and
\begin{equation}
\operatorname{PoF}_{N,LN}
:=
1-\frac{\operatorname{FAIR}_{N,LN}}{S_{N,LN}}.
\end{equation}
Then
\begin{equation}
\operatorname{PoF}_{N,LN}\xrightarrow{\mathrm{P}}0.
\end{equation}
\end{corollary}

\begin{proof}[Proof of Corollary~\ref{cor:vanishing-pof}]
Since
\begin{equation}
t_i\left(g_{\mathrm{mm}};\mU_{(N,LN)}\right)
\geq M_{N,LN}
\end{equation}
for every agent \(i\),
\begin{equation}
\operatorname{FAIR}_{N,LN}\geq NM_{N,LN}.
\end{equation}
Moreover, since \(S_{N,LN}\) is the maximum sum-welfare value,
\begin{equation}
S_{N,LN}\geq\operatorname{FAIR}_{N,LN}.
\end{equation}
Therefore,
\begin{align}
\label{eq:POFproduct}
0
\leq
\operatorname{PoF}_{N,LN}
\leq
1-\frac{NM_{N,LN}}{S_{N,LN}}.
\end{align}
Using \eqref{eq:S_N-quantile},
Theorem~\ref{THM:main_theorem_order_statistic}, and
Lemma~\ref{lemma: Asymptotic equivalence of ratio quantiles}, we obtain
\begin{align}
\label{eq:Mnproduct}
\frac{NM_{N,LN}}{S_{N,LN}}
\xrightarrow{\mathrm{P}}1.
\end{align}
The result now follows from \eqref{eq:POFproduct} by the squeeze
theorem.
\end{proof}

We now turn to the proof of the main theorem. We begin with the lower
bound. Under the more restrictive admissible Weibull-type assumption,
the relative lower bound can be strengthened to an additive one.
Although this strengthening is not needed for the proof of
Theorem~\ref{THM:main_theorem_order_statistic}, it provides a sharper
asymptotic characterization and is therefore stated as a separate
result.

\begin{restatable}{lemma}{lowerbound}
\label{theorem:main-first-model}
Suppose that \(F\) has a regularly varying upper-tail quantile or a
finite positive upper endpoint. Then, for every \(\gre>0\),
\begin{align}
\lim_{N\to\infty}
\mathbb{P}\left(
1-
\frac{M_{N,LN}}
{L F^{-1}\left(1-\frac{\ln N}{N}\right)}
>\gre
\right)
=0.
\tag{a}
\end{align}
If \(F\) has an admissible Weibull-type tail or a finite positive
upper endpoint, then, for every \(\gre>0\),
\begin{align}
\lim_{N\to\infty}
\mathbb{P}\left(
L F^{-1}\left(1-\frac{\ln N}{N}\right)
-
M_{N,LN}
>\gre
\right)
=0.
\tag{b}
\end{align}
\end{restatable}

\begin{proof}
See appendix~\ref{app:Proportional Growth: Lower Bound}.
\end{proof}

The following lemma provides the corresponding upper bound.

\begin{restatable}{lemma}{upperbound}
\label{theorem:one-sided-upper-bound-max-min}
Suppose that \(F\) has a Weibull-type tail or a finite positive upper
endpoint. Then, for every \(\gre>0\),
\begin{equation}
\lim_{N\to\infty}
\mathbb{P}\left(
1-
\frac{
L F^{-1}\left(1-\frac{1}{N}\right)
}{
M_{N,LN}
}
>\gre
\right)
=0.
\end{equation}
\end{restatable}

\begin{proof}
See appendix~\ref{app:Proportional Growth: Upper Bound}.
\end{proof}

To relate the two bounds, we use the following quantile equivalence.

\begin{restatable}{lemma}{weibull}
\label{lemma: Asymptotic equivalence of ratio quantiles}
Suppose that \(F\) has a Weibull-type tail, as defined in
Definition~\ref{def:Weibull-type-tail}, or a finite positive upper
endpoint. Let
\begin{equation}
\gz_1^N
:=
\frac{\ln N+\ln\ln\ln N}{N}.
\end{equation}
Then
\begin{equation}
\frac{
F^{-1}\left(1-\frac{1}{N}\right)
}{
F^{-1}\left(1-\gz_1^N\right)
}
\longrightarrow1
\qquad\text{as }N\to\infty.
\end{equation}
\end{restatable}

\begin{proof}
See appendix~\ref{app:first-model}.
\end{proof}

This leads to the following conclusion: for Weibull-type utility
distributions, \(M_{N,LN}\) is asymptotically equivalent in ratio to
\begin{equation}
L F^{-1}\left(1-\frac{\ln N}{N}\right).
\end{equation}
Theorem~\ref{THM:main_theorem_order_statistic} formalizes this
conclusion.

\begin{proof}[Proof of Theorem~\ref{THM:main_theorem_order_statistic}]
Define
\begin{equation}
q_N
:=
F^{-1}\left(1-\frac{\ln N}{N}\right),
\qquad
r_N
:=
F^{-1}\left(1-\frac{1}{N}\right).
\end{equation}
Also, let
\begin{equation}
\widetilde q_N
:=
F^{-1}\left(1-\frac{\ln N+\ln\ln\ln N}{N}\right).
\end{equation}
By monotonicity of \(F^{-1}\),
\begin{equation}
\widetilde q_N\leq q_N\leq r_N.
\end{equation}
Lemma~\ref{lemma: Asymptotic equivalence of ratio quantiles} gives
\begin{equation}
\frac{r_N}{\widetilde q_N}\longrightarrow1,
\end{equation}
and therefore
\begin{equation}
\frac{r_N}{q_N}\longrightarrow1.
\end{equation}

Lemma~\ref{theorem:main-first-model} gives, for every \(\gre>0\),
\begin{equation}
\mathbb{P}\left(
\frac{M_{N,LN}}{Lq_N}<1-\gre
\right)
\longrightarrow0.
\end{equation}
Moreover, Lemma~\ref{theorem:one-sided-upper-bound-max-min}, together
with \(r_N/q_N\to1\), gives
\begin{equation}
\mathbb{P}\left(
\frac{M_{N,LN}}{Lq_N}>1+\gre
\right)
\longrightarrow0.
\end{equation}
Combining the two one-sided bounds yields
\begin{equation}
\frac{M_{N,LN}}{Lq_N}
\xrightarrow{\mathrm{P}}1.
\end{equation}
Equivalently,
\begin{equation}
\lim_{N\to\infty}
\mathbb{P}\left(
\left|
1-
\frac{M_{N,LN}}
{L F^{-1}\left(1-\frac{\ln N}{N}\right)}
\right|>\gre
\right)
=0.
\end{equation}
\end{proof}

\section{Fixed-Agent Model: Many Resources Per Agent}
\label{sec:second_model}

Another interesting regime for max-min fairness with i.i.d. utilities
arises when the number of resources is much larger than the number of
agents. In this section, we examine the max-min value when \(N\) is
fixed and \(K\to\infty\).

In this regime, the previously established matching approach becomes
inapplicable. If each agent is allocated approximately \(\frac{K}{N}\) of its
highest-valued resources, the corresponding cutoff is approximately
the \(\frac{K}{N}\)-th largest order statistic among \(K\) observations,  Consequently, we adopt a strategy in
which every resource is assigned to the agent with the highest utility.
This strategy coincides with the opportunistic optimal sum allocation.

We show that, under this allocation, every agent receives a sufficient
number of resources for the achieved minimum utility to be
asymptotically optimal for the max-min objective.

\begin{theorem}
\label{thm:fixed-agent}
Let \(N\) be fixed and let \(K\to\infty\). Let
\begin{equation}
\mU_{(N,K)}
=
(U_{i,j})_{\substack{1\leq i\leq N\\1\leq j\leq K}}
\end{equation}
be the random utility matrix defined in the problem formulation, and
suppose that
\begin{equation}
\mathbb{E}[U_{1,1}]<\infty,
\qquad
\operatorname{Var}(U_{1,1})<\infty.
\end{equation}
Define
\begin{equation}
Y_j:=\max_{1\leq i\leq N}U_{i,j},
\qquad
\mu_U^N:=\mathbb{E}[Y_j],
\qquad
v_U^N:=\operatorname{Var}(Y_j)<\infty.
\end{equation}
Then, for every \(\delta>0\),
\begin{align}
\mathbb{P}\left(
\left|
\frac{M_{N,K}}{\frac{K}{N}\mu_U^N}-1
\right|>\delta
\right)
&\leq
\frac{v_U^N}
{K\delta^2(\mu_U^N)^2}
\nonumber\\
&+
\frac{
N\left[Nv_U^N+(N-1)(\mu_U^N)^2\right]
}{
K\delta^2(\mu_U^N)^2
}
\xrightarrow[K\to\infty]{}0.
\label{eq:corrected-fixed-agent-bound}
\end{align}
Consequently,
\begin{equation}
\frac{M_{N,K}}
     {\left(\frac{K}{N}\right)\mu_U^N}
\xrightarrow{\mathrm{P}}1.
\end{equation}
\end{theorem}

\begin{proof}
For every resource \(j\), let
\begin{equation}
I_j:=\arg\max_{1\leq i\leq N}U_{i,j}.
\end{equation}
Because the utilities are i.i.d. and continuously distributed,
\(I_j\) is uniformly distributed on \(\{1,\ldots,N\}\) and is
independent of \(Y_j\).

Assign every resource to its maximizing agent and define
\begin{equation}
G_i(K):=
\sum_{j=1}^K
Y_j\mathbf{1}_{\{I_j=i\}}.
\end{equation}
This is a feasible allocation. Therefore,
\begin{equation}
\min_{1\leq i\leq N}G_i(K)
\leq M_{N,K}
\leq \frac{S_{N,K}}{N},
\qquad
S_{N,K}=\sum_{j=1}^K Y_j.
\end{equation}

Set
\begin{equation}
Z_{i,j}:=Y_j\mathbf{1}_{\{I_j=i\}}.
\end{equation}
The independence of \(I_j\) and \(Y_j\) gives
\begin{equation}
\mathbb{E}[Z_{i,j}]
=
\frac{\mu_U^N}{N}.
\end{equation}
Moreover,
\begin{align*}
\operatorname{Var}(Z_{i,j})
&=
\frac{\mathbb{E}[Y_j^2]}{N}
-\frac{(\mu_U^N)^2}{N^2}\\
&=
\frac{v_U^N}{N}
+\frac{N-1}{N^2}(\mu_U^N)^2.
\end{align*}

For the upper deviation, using
\(M_{N,K}\leq S_{N,K}/N\) and Chebyshev's inequality,
\begin{align*}
&\mathbb{P}\left(
M_{N,K}>
(1+\delta)\frac{K\mu_U^N}{N}
\right)\\
&\quad\leq
\mathbb{P}\left(
S_{N,K}-K\mu_U^N
>
\delta K\mu_U^N
\right)\\
&\quad\leq
\frac{v_U^N}
{K\delta^2(\mu_U^N)^2}.
\end{align*}

For the lower deviation,
\begin{align*}
&\mathbb{P}\left(
M_{N,K}<
(1-\delta)\frac{K\mu_U^N}{N}
\right)\\
&\quad\leq
\mathbb{P}\left(
\min_{1\leq i\leq N}G_i(K)
<
(1-\delta)\frac{K\mu_U^N}{N}
\right)\\
&\quad\leq
\sum_{i=1}^N
\mathbb{P}\left(
G_i(K)-\frac{K\mu_U^N}{N}
<
-\delta\frac{K\mu_U^N}{N}
\right)\\
&\quad\leq
\frac{
N\left[Nv_U^N+(N-1)(\mu_U^N)^2\right]
}{
K\delta^2(\mu_U^N)^2
}.
\end{align*}

Combining the upper and lower deviations proves
\eqref{eq:corrected-fixed-agent-bound}. Since \(N\) is fixed,
the right-hand side converges to zero as \(K\to\infty\).
\end{proof}

We next show that the optimal sum welfare is asymptotically equivalent
to \(N\) times the max-min value. Fix \(0<\delta<1\) and define
\begin{equation}
r_\delta:=\sqrt{1-\delta}.
\end{equation}
In particular,
\begin{equation}
0<r_\delta<1,
\qquad
r_\delta^2=1-\delta.
\end{equation}
We have the event inclusion
\begin{align}
\left\{
\frac{NM_{N,K}}{S_{N,K}}\leq1-\delta
\right\}&
\subseteq{}
\left\{
\frac{NM_{N,K}}{K\mu_U^N}
\leq r_\delta
\right\}\\
&\cup
\left\{
\frac{S_{N,K}}{K\mu_U^N}
\geq\frac{1}{r_\delta}
\right\}.
\end{align}
Indeed, outside the two events on the right-hand side,
\begin{equation}
\frac{NM_{N,K}}{S_{N,K}}
=
\frac{
\frac{NM_{N,K}}{K\mu_U^N}
}{
\frac{S_{N,K}}{K\mu_U^N}
}
>
\frac{r_\delta}{r_\delta^{-1}}
=
r_\delta^2
=
1-\delta.
\end{equation}

Therefore,
\begin{align*}
&\mathbb{P}\left(
\frac{NM_{N,K}}{S_{N,K}}\leq1-\delta
\right)\\
&\quad\leq
\mathbb{P}\left(
\frac{NM_{N,K}}{K\mu_U^N}
\leq r_\delta
\right)
+
\mathbb{P}\left(
\frac{S_{N,K}}{K\mu_U^N}
\geq\frac{1}{r_\delta}
\right)\\
&\quad\leq
\frac{
N\left[Nv_U^N+(N-1)(\mu_U^N)^2\right]
}{
K(1-r_\delta)^2(\mu_U^N)^2
}
+
\frac{
v_U^N
}{
K(r_\delta^{-1}-1)^2(\mu_U^N)^2
}\\
&\quad=
\frac{
N\left[Nv_U^N+(N-1)(\mu_U^N)^2\right]
}{
K\left(1-\sqrt{1-\delta}\right)^2(\mu_U^N)^2
}
\\&+
\frac{
v_U^N
}{
K\left((1-\delta)^{-1/2}-1\right)^2(\mu_U^N)^2
}
\xrightarrow[K\to\infty]{}0.
\end{align*}

Since
\begin{equation}
0\leq\frac{NM_{N,K}}{S_{N,K}}\leq1,
\end{equation}
the upper deviation from \(1\) is impossible. The preceding bound
therefore proves that
\begin{equation}
\frac{NM_{N,K}}{S_{N,K}}
\xrightarrow{\mathrm{P}}1.
\end{equation}
Equivalently,
\begin{equation}
\frac{S_{N,K}}{NM_{N,K}}
\xrightarrow{\mathrm{P}}1.
\end{equation}

Thus, the optimal sum welfare is asymptotically equivalent to
\(N\) times the max-min value. As a consequence, the price of
fairness also converges to zero.

\begin{corollary}
\label{corollary:fixed-agent-pof}
Under the assumptions of Theorem~\ref{thm:fixed-agent}, let
\begin{equation}
g_{\mathrm{mm}}
\in
\mathcal G_{\mathrm{mm}}\left(\mU_{(N,K)}\right)
\end{equation}
be any max-min allocation of the random utility matrix and define
\begin{align}
\operatorname{FAIR}_{N,K}
&:=
\sum_{i=1}^{N}
t_i\left(g_{\mathrm{mm}};\mU_{(N,K)}\right),
\\
\operatorname{PoF}_{N,K}
&:=
1-\frac{\operatorname{FAIR}_{N,K}}{S_{N,K}}.
\end{align}
Then, for fixed \(N\) and \(K\to\infty\),
\begin{equation}
\operatorname{PoF}_{N,K}\xrightarrow{\mathrm{P}}0.
\end{equation}
\end{corollary}

\begin{proof}
Since
\begin{equation}
t_i\left(g_{\mathrm{mm}};\mU_{(N,K)}\right)
\geq M_{N,K}
\end{equation}
for every agent \(i\),
\begin{equation}
\operatorname{FAIR}_{N,K}\geq NM_{N,K}.
\end{equation}
Moreover, because \(S_{N,K}\) is the maximum achievable total
utility,
\begin{equation}
\operatorname{FAIR}_{N,K}\leq S_{N,K}.
\end{equation}
Therefore,
\begin{equation}
0
\leq
\operatorname{PoF}_{N,K}
\leq
1-\frac{NM_{N,K}}{S_{N,K}}.
\end{equation}
By the ratio convergence established above,
\begin{equation}
\frac{NM_{N,K}}{S_{N,K}}
\xrightarrow{\mathrm{P}}1.
\end{equation}
The result now follows from the squeeze theorem.
\end{proof}
\section{Numerical Experiments}
\label{sec:simulations}

In this section, we numerically examine the finite-size behavior of the
max-min value in the three scaling regimes considered in the paper. The
unnormalized percentile plots are presented for three representative utility
distributions: exponential ($\lambda=1$), Nakagami ($m=1, \Omega=1$), and Rician ($\nu=1, \sigma=1$). The additional normalized
convergence plots broaden the comparison to Weibull($\beta=\frac{1}{2})$, uniform $[0,1]$, Pareto type II with tail index $\alpha=3$, and  Lognormal ($\mu=0,\sigma=1$) distributions.

For every simulated choice of \(N\) and \(K\), we generated \(1000\)
independent realizations \(\mathbf u_{(N,K)}\) of \(\mU_{(N,K)}\). For each
simulated quantity, we report its \(5\)th percentile, median, and \(95\)th
percentile over the Monte Carlo trials. The median represents the typical observed value, while the \(5\)th and \(95\)th percentiles describe its
finite-size variability. When the exact max-min value cannot be computed efficiently (due to a very low convergence rate or the NP hardness of the problem), these percentiles are reported separately for an empirical lower
bound and an empirical upper bound.

\subsection{A Single Resource per Agent $(N=K)$}
$(N=K)$. When the number of agents equals the number of resources, the max-min allocation problem reduces to a bottleneck assignment problem and can
therefore be solved in polynomial time. For every simulated realization
\(\mathbf u_{(N,N)}\), we compute the exact max-min value
\(m_{N,N}(\mathbf u_{(N,N)})\) by binary searching over the distinct utility
thresholds and testing matching feasibility at each threshold using the Hopcroft-Karp algorithm.

Figure~\ref{fig:simulations-balanced} presents the \(5\)th percentile, median,
and \(95\)th percentile of the empirical max-min value as functions of \(N\).
The empirical values are compared with the theoretical centering value
\begin{equation}
    q_N
    :=
    F^{-1}\left(1-\frac{\ln N}{N}\right).
    \label{eq:sim-balanced-centering}
\end{equation}
The three panels correspond to the exponential, Nakagami, and Rician distributions. For all three distributions, the empirical max-min value
approaches the theoretical prediction as \(N\) increases, consistent with
Theorem~\ref{theorem:main}.

To compare the convergence rate across a broader collection of distributions,
Figure~\ref{fig:simulations-balanced-normalized} plots the normalized empirical
median
\begin{equation}
    \frac{\operatorname{median}(M_{N,N})}
    {F^{-1}\left(1-\frac{\ln N}{N}\right)}
    \label{eq:sim-balanced-normalized-median}
\end{equation}
for Weibull, uniform, Nakagami, Rician, lognormal, Pareto type II, and exponential utilities. All these distributions have regularly varying upper-tail quantiles or are bounded from above. Therefore, Theorem~\ref{theorem:main} applies to each of them and guarantees that, in the balanced setting, the corresponding normalized max-min value converges to 1 in probability. The figure illustrates and compares the empirical convergence rates predicted by the theorem.

\begin{figure*}[htbp]
    \centering

    \begin{subfigure}[t]{0.32\textwidth}
        \centering
        \includegraphics[width=\linewidth]
        {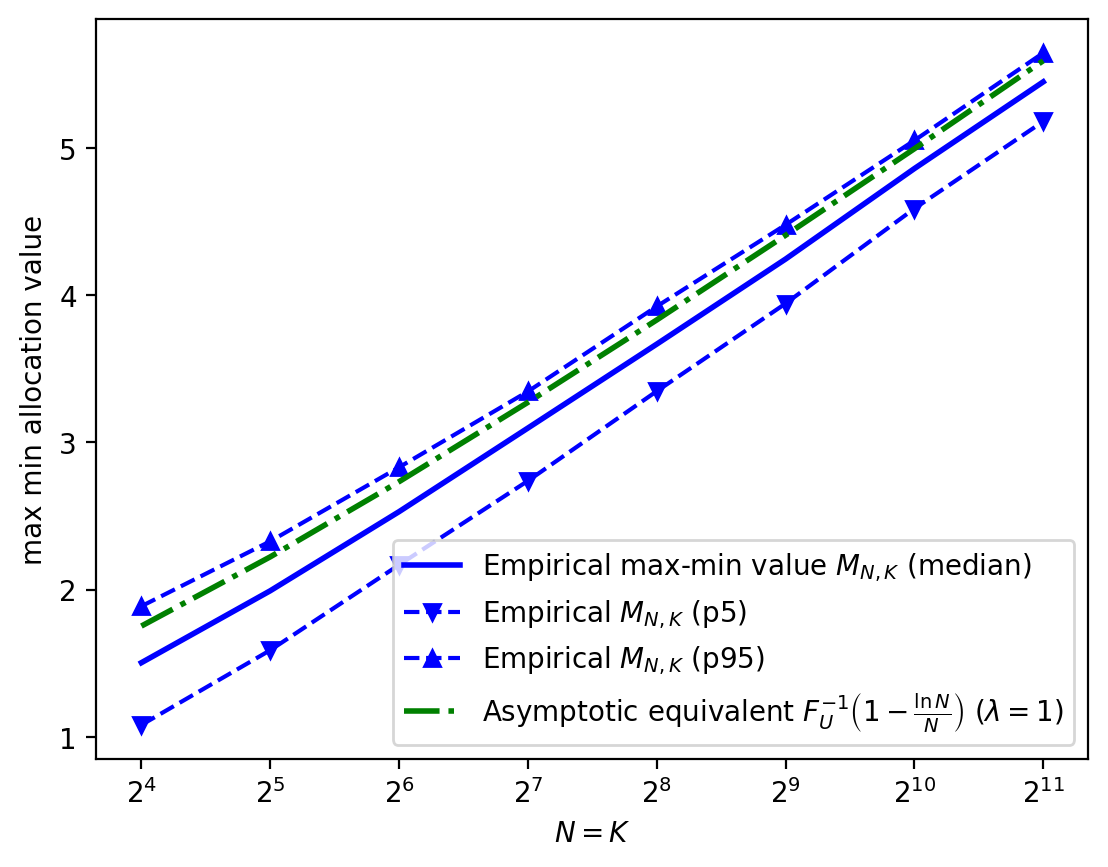}
        \caption{Exponential distribution (\(\lambda=1\)).}
        \label{fig:balanced-exponential}
    \end{subfigure}
    \hfill
    \begin{subfigure}[t]{0.32\textwidth}
        \centering
        \includegraphics[width=\linewidth]
        {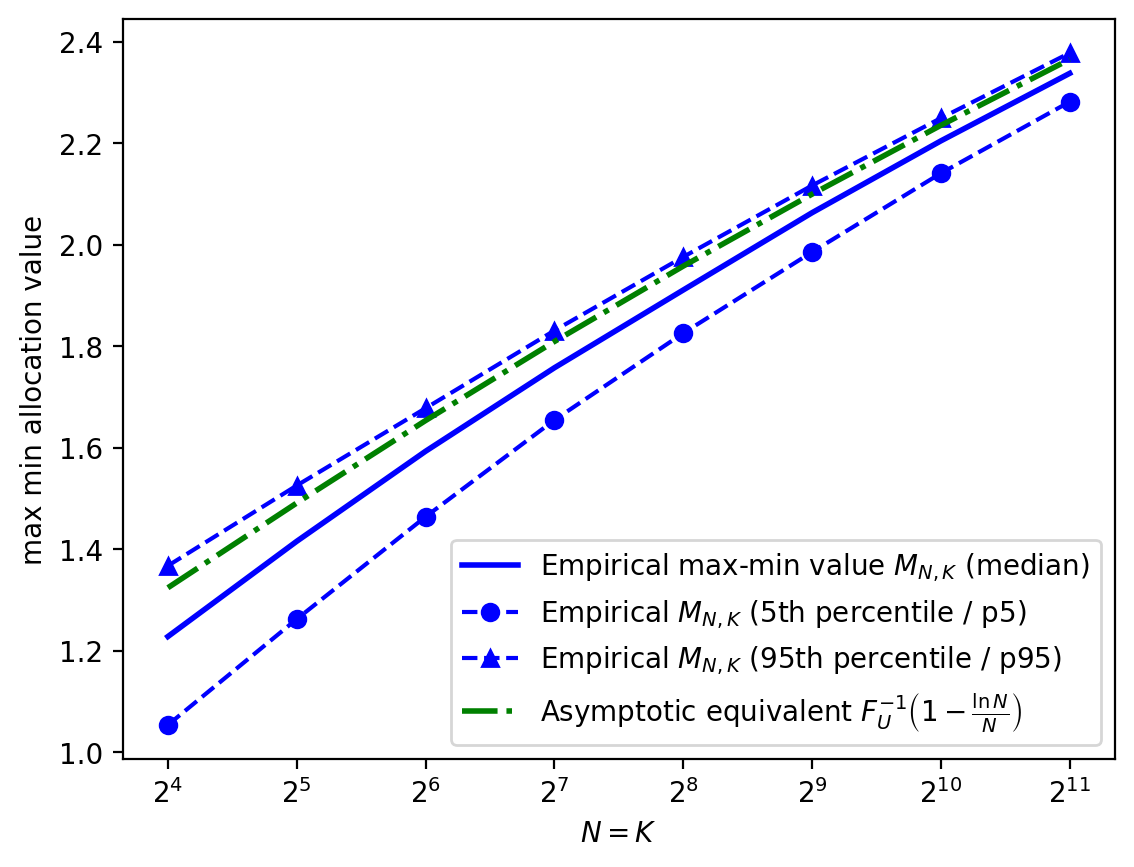}
        \caption{Nakagami distribution (\(m=1\), \(\Omega=1\)).}
        \label{fig:balanced-nakagami}
    \end{subfigure}
    \hfill
    \begin{subfigure}[t]{0.32\textwidth}
        \centering
        \includegraphics[width=\linewidth]
        {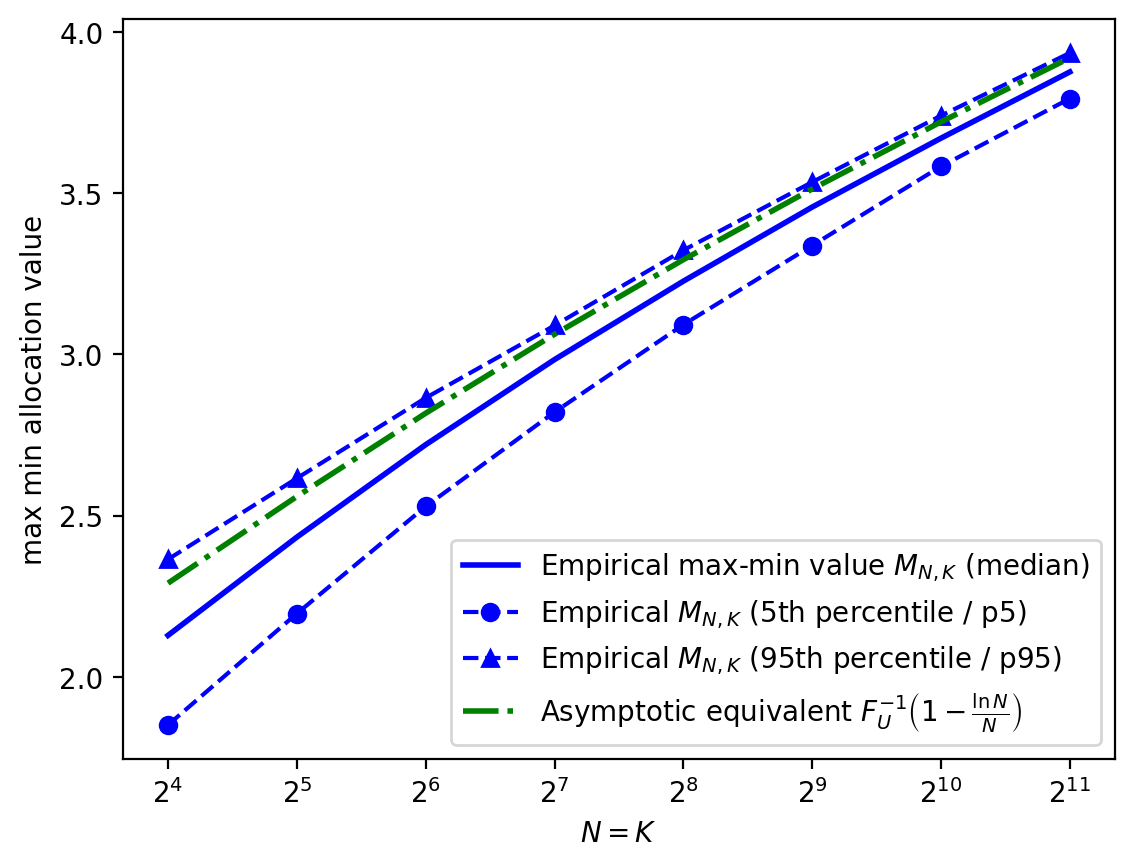}
        \caption{Rician distribution (\(\nu=1\), \(\sigma=1\)).}
        \label{fig:balanced-rice}
    \end{subfigure}

    \caption{Empirical max-min value in the balanced regime \(N=K\).
    Each panel presents the \(5\)th percentile, median, and \(95\)th
    percentile over \(1{,}000\) Monte Carlo trials, together with the
    theoretical centering value
    \(F^{-1}(1-\frac{\ln N}{N})\).}
    \label{fig:simulations-balanced}
\end{figure*}

\begin{figure*}[htbp]
    \centering
    \includegraphics[width=0.72\textwidth]
    {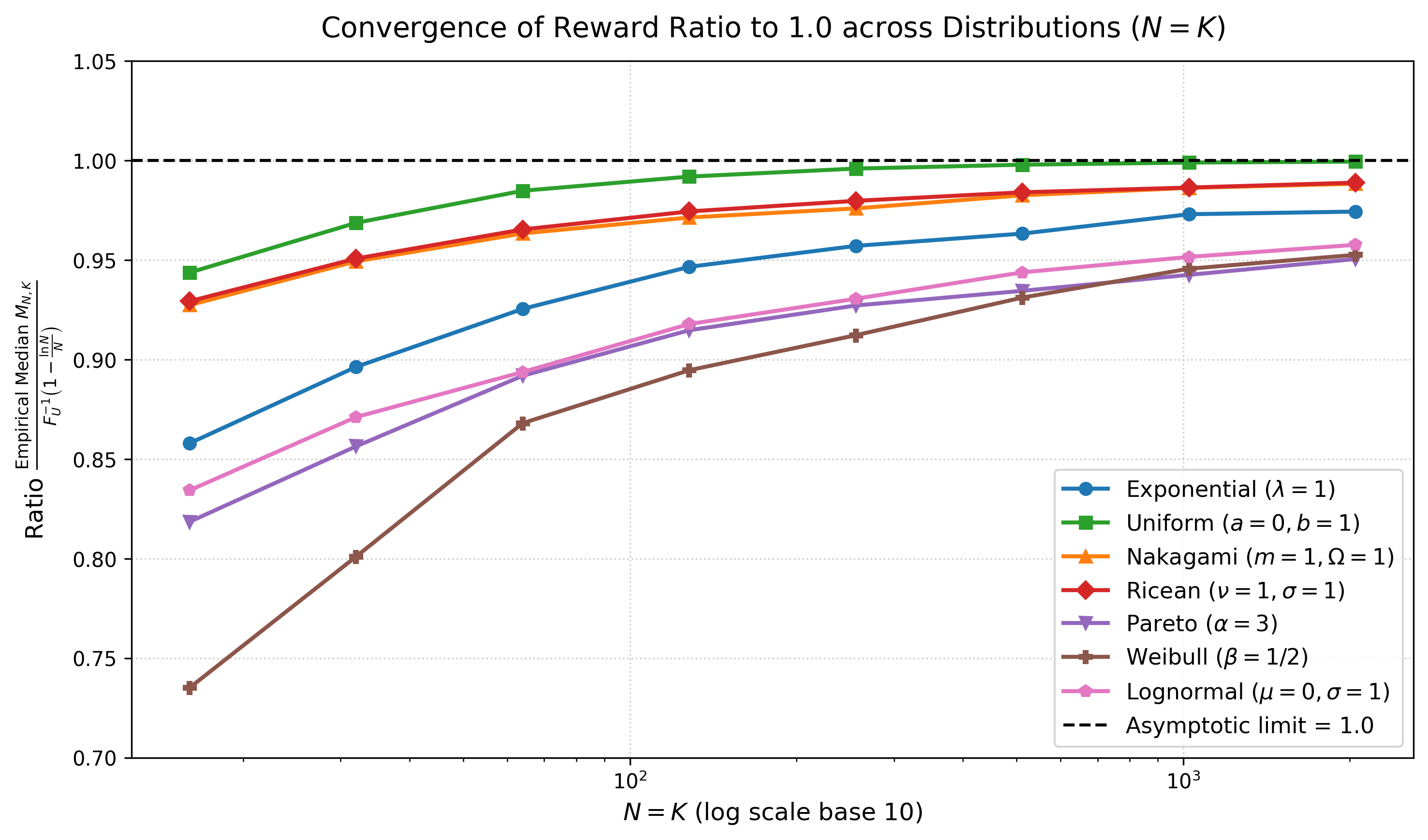}
    \caption{Normalized empirical median in the balanced regime \(N=K\).
    Each curve shows
    \(\operatorname{median}(M_{N,N})/
    F^{-1}(1-\ln N/N)\) over \(1{,}000\) Monte Carlo trials. The
    horizontal reference line is \(1\).}
    \label{fig:simulations-balanced-normalized}
\end{figure*}

\subsection{Proportional Growth Model (\(K=LN\))}

Next, we consider the proportional-growth regime \(K=LN\), where \(L\) is
fixed. Let
\begin{equation}
    q_N:=F^{-1}\left(1-\frac{\ln N}{N}\right),\ \ 
    b_N:=F^{-1}\left(1-\frac{1}{N}\right),
\end{equation}
so that the max-min lower scale is \(Lq_N\), whereas the sum-welfare upper
scale is \(Lb_N\).

The empirical lower bound is constructed by partitioning the \(K=LN\)
resources into \(L\) disjoint blocks
\(\mathcal{B}_1,\ldots,\mathcal{B}_L\), each containing \(N\) resources.
For every block, we compute the exact max-min value of the corresponding
\(N\times N\) assignment problem:
\begin{equation}
    M_N^{(\ell)}
    :=
    \max_{\pi_\ell:[N]\to\mathcal{B}_\ell\ {\rm bijective}}
    \min_{i\in[N]} U_{i,\pi_\ell(i)} .
\end{equation}
Combining the \(L\) blockwise matchings gives a feasible allocation for the
original problem and hence the lower bound
\begin{equation}
    \underline{M}_{N,LN}
    :=
    \sum_{\ell=1}^{L} M_N^{(\ell)}
    \leq M_{N,LN}.
\end{equation}

The upper bound is obtained by assigning each resource to the agent who
values it most and then dividing the resulting sum welfare by \(N\):
\begin{equation}
    \overline{M}_{N,LN}
    :=
    \frac{1}{N}\sum_{j=1}^{LN}\max_{i\in[N]}U_{i,j}.
\end{equation}
Indeed, the minimum utility of any allocation cannot exceed its average
utility across agents, while its total utility cannot exceed the sum obtained
by assigning every resource to its highest-valuing agent. If
\(U_{[N:N]}^{(j)}:=\max_{i\in[N]}U_{i,j}\), then
\begin{equation}
    \overline{M}_{N,LN}
    =
    L\left(
        \frac{1}{LN}\sum_{j=1}^{LN}U_{[N:N]}^{(j)}
    \right).
\end{equation}
Thus, the upper bound is \(L\) times the empirical average of \(K=LN\)
independent \(N\)-sample maxima.

To separate the two sources of finite-size error, we consider the three ratios
\begin{align}
    R_N^{\mathrm{lb}}
    &:=
    \frac{\operatorname{median}(\underline{M}_{N,LN})}{Lq_N},
    &
    R_N^{\mathrm{ub}}
    &:=
    \frac{\operatorname{median}(\overline{M}_{N,LN})}{Lb_N},
    \label{eq:sim-proportional-bound-ratios}\\
    R_N^{\mathrm{q}}
    &:=
    \frac{q_N}{b_N}
    =
    \frac{F^{-1}\left(1-\frac{\ln N}{N}\right)}
         {F^{-1}\left(1-\frac{1}{N}\right)}.
    \label{eq:sim-proportional-quantile-ratio}
\end{align}
The first two ratios measure the convergence of each bound to its corresponding
finite-size scale, while the third measures the remaining separation between
these two scales. Under the assumptions of
Theorem~\ref{THM:main_theorem_order_statistic},
\begin{equation}
    R_N^{\mathrm{lb}}\longrightarrow 1,
    \qquad
    R_N^{\mathrm{ub}}\longrightarrow 1,
    \qquad
    R_N^{\mathrm{q}}\longrightarrow 1.
\end{equation}
Thus, for the distributions covered by the theorem, the result guarantees
convergence in ratio, which is precisely the mode of convergence examined by
these normalized quantities.

Figure~\ref{fig:simulations-proportional} presents these three complementary
comparisons for Weibull, uniform, Nakagami, Rician, and exponential utilities.
\footnote{Pareto type II and lognormal are unbounded and do not have Weibull-type upper-tail quantiles required by  Theorem~\ref{THM:main_theorem_order_statistic}}. 
For these distributions we expect
\begin{equation}
    \frac{M_{N,LN}}{Lq_N}
    \xrightarrow{\mathbb{P}}1.
\end{equation}
The lower-bound computation becomes prohibitively expensive because convergence is very slow. Therefore, the first
panel is restricted to problem sizes for which the bound can be evaluated reliably. The second panel shows that the upper bound, normalized by \(Lb_N\),
approaches one. To examine substantially larger values of \(N\) without
generating and optimizing random utility matrices, the third panel plots the
deterministic ratio \(R_N^{\mathrm q}=q_N/b_N\). Taken together, the panels
show that the lower and upper bounds approach their respective finite-size
scales and that these scales themselves become asymptotically equivalent.

For the unbounded distributions considered here, \(R_N^{\mathrm q}\) may
approach one remarkably slowly. 
\begin{figure*}[htbp]
    \centering

    \begin{subfigure}[t]{0.32\textwidth}
        \centering
        \includegraphics[width=\linewidth]
        {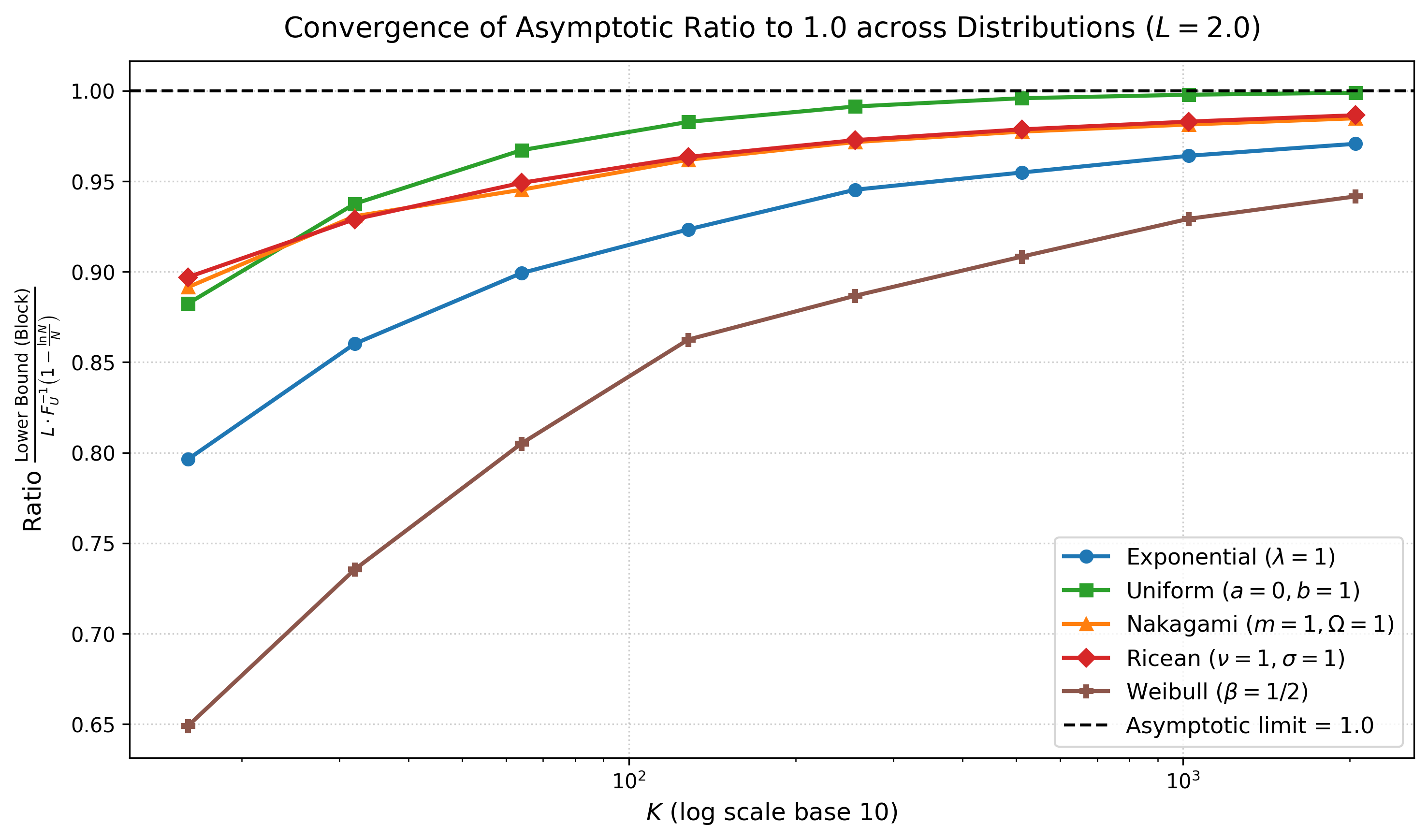}
        \caption{Lower bound normalized by \(Lq_N\).}
        \label{fig:proportional-lower-ratio}
    \end{subfigure}
    \hfill
    \begin{subfigure}[t]{0.32\textwidth}
        \centering
        \includegraphics[width=\linewidth]
        {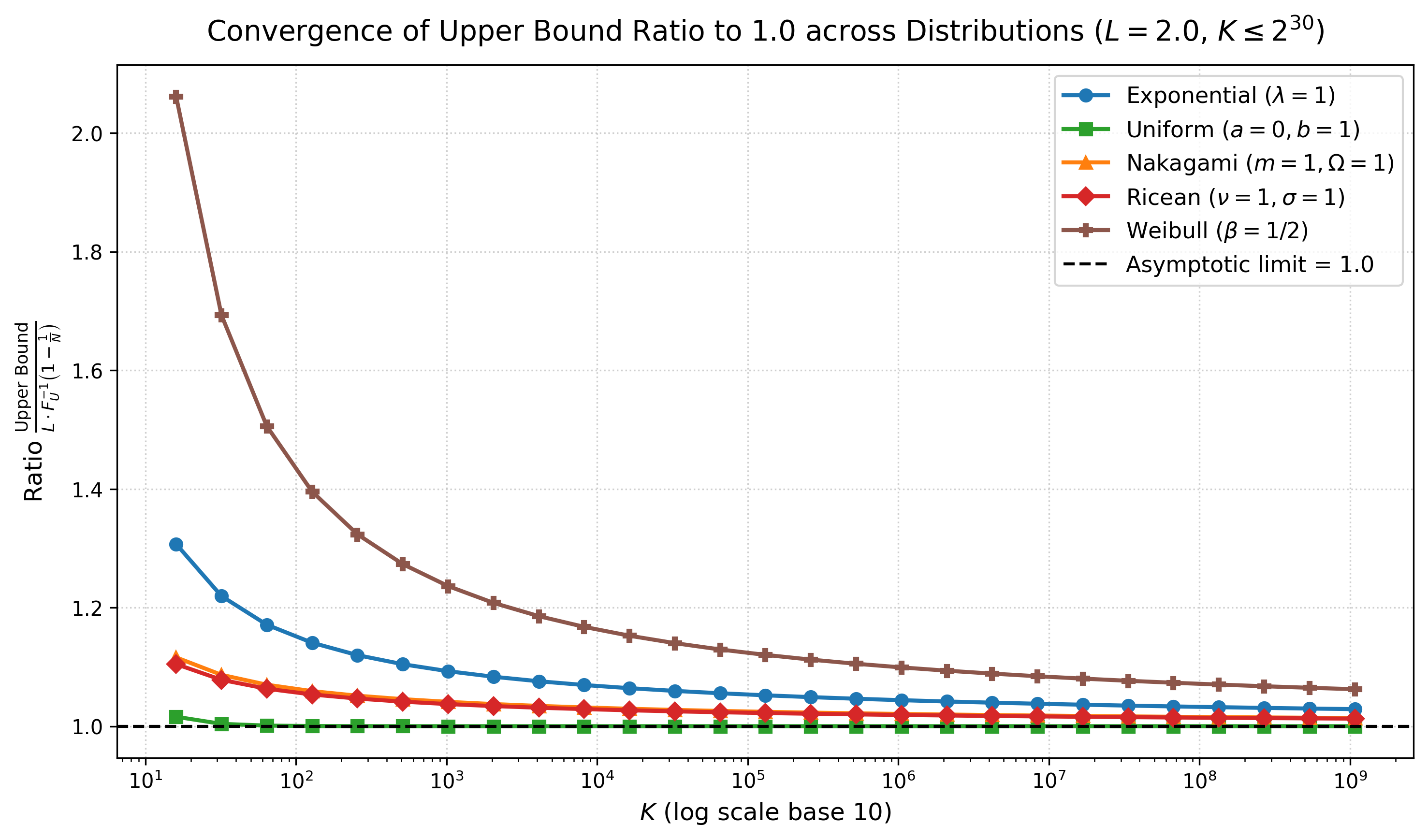}
        \caption{Upper bound normalized by \(Lb_N\).}
        \label{fig:proportional-upper-ratio}
    \end{subfigure}
    \hfill
    \begin{subfigure}[t]{0.32\textwidth}
        \centering
        \includegraphics[width=\linewidth]
        {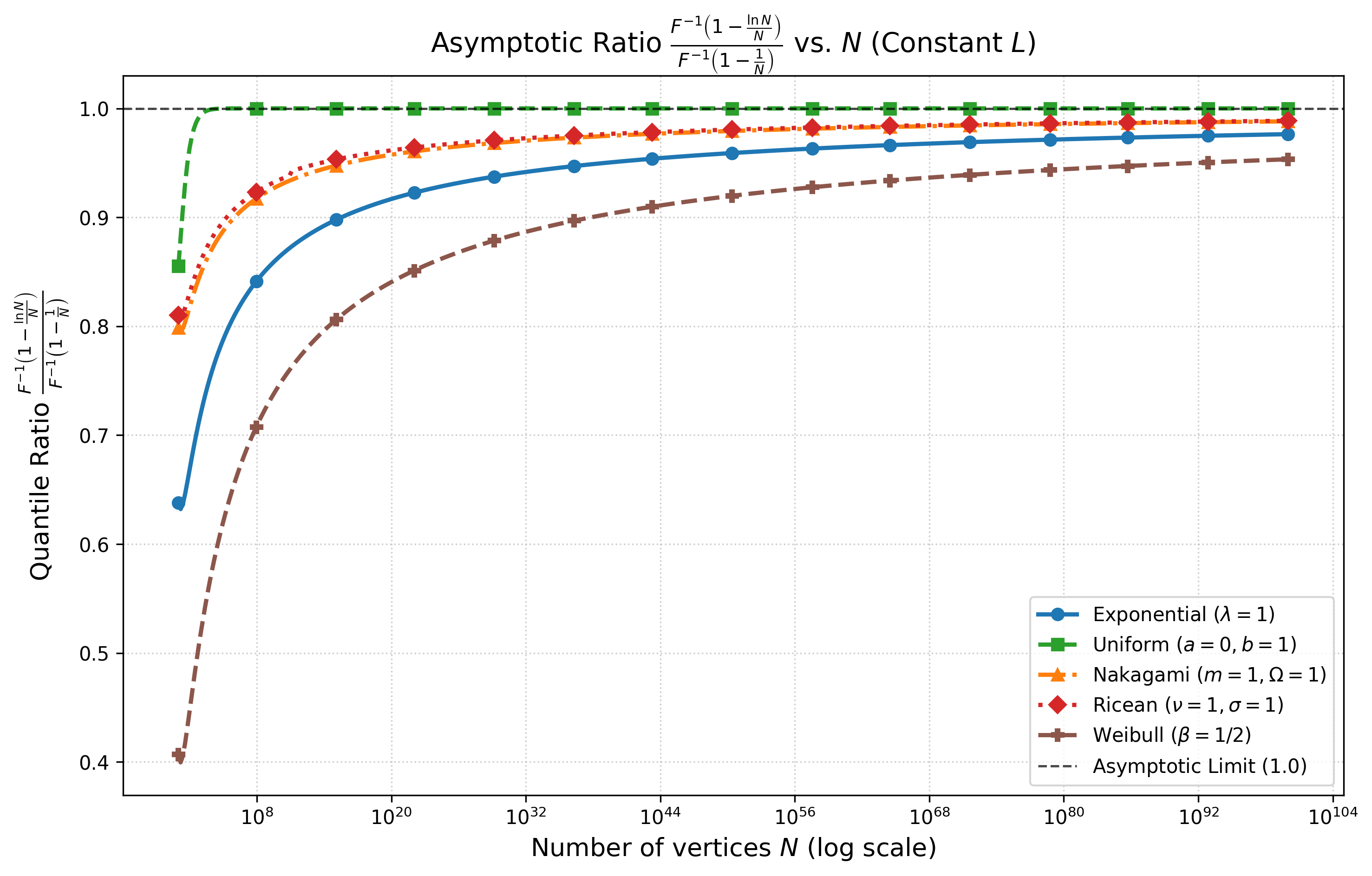}
        \caption{Deterministic quantile ratio \(q_N/b_N\).}
        \label{fig:proportional-quantile-ratio}
    \end{subfigure}

    \caption{Convergence diagnostics in the proportional-growth regime
    \(K=LN\) for the Weibull, uniform, Nakagami, Rician, and exponential
    distributions. The first two panels show the normalized empirical medians
    \(\operatorname{median}(\underline{M}_{N,LN})/(Lq_N)\) and
    \(\operatorname{median}(\overline{M}_{N,LN})/(Lb_N)\), respectively,
    over \(1{,}000\) Monte Carlo trials. The third panel shows the deterministic
    ratio \(q_N/b_N\) over a much larger range of \(N\). The horizontal
    reference line in each panel is \(1\).}
    \label{fig:simulations-proportional}
\end{figure*}

\subsection{Fixed-Agent Model (\(K\gg N\))}
Finally, we consider the fixed-agent regime, in which \(N\) remains fixed
while \(K\) tends to infinity. We consider \(N=2\) and \(N=32\) and evaluate
the empirical bounds for increasing values of \(K\). Including \(N=32\)
examines whether the fixed-agent asymptotic behavior is already visible for a
substantially larger, but still fixed, number of agents. Throughout this
subsection, for \(N=2\) we use
\(K\in\{2^4,2^5,\ldots,2^{11}\}\), whereas for \(N=32\) we use
\(K\in\{2^6,2^7,\ldots,2^{11}\}\).

For the empirical upper bound, we use the standard linear-programming relaxation of the max-min fair allocation problem presented in \cite{zehavi2013weighted}. For a realization $u=(u_{i,j})$, introduce allocation variables $x_{i,j}$ and a common utility level $t$, and define
\begin{align}
\overline{m}_{N,K}(u)
&:= \max_{x,t}\ t \\
\text{s.t.}\qquad
\sum_{j=1}^{K} u_{i,j}x_{i,j}
&\geq t,
&& i=1,\ldots,N, \nonumber\\
\sum_{i=1}^{N} x_{i,j}
&=1,
&& j=1,\ldots,K, \nonumber\\
0\leq x_{i,j}
&\leq 1,
&& i=1,\ldots,N,\quad j=1,\ldots,K.
\nonumber
\end{align}
The original indivisible-allocation problem is recovered by imposing
$x_{i,j}\in\{0,1\}$. Hence the continuous relaxation enlarges the feasible
set and satisfies
\begin{equation}
m_{N,K}(u)\leq \overline{m}_{N,K}(u).
\end{equation}
For the random utility matrix, we write
$\overline{M}_{N,K}:=\overline{m}_{N,K}(U_{(N,K)})$. \newline\newline
For the empirical lower bound, we use the forest-rounding procedure of
\cite{zehavi2013weighted}, formulated through the bipartite support graph.
Let \(x^\star\) be an optimal solution of the linear-programming relaxation.
Its support graph has one vertex for each agent and one vertex for each
resource, with an edge between agent \(i\) and resource \(j\) whenever
\(x^\star_{i,j}>0\). If this graph contains a cycle, the allocation fractions
along the cycle are adjusted while preserving feasibility and every agent's
total fractional utility, until at least one positive fraction becomes zero.
Repeating this cycle-elimination procedure produces an optimal fractional
allocation whose bipartite support graph is a forest.

Root each tree at an arbitrary agent and assign every resource entirely to
its parent agent. Resources that are already assigned integrally retain
their assignments, since their vertices are leaves. Each non-root agent
loses only its fractional share of its parent resource, while receiving
all its child resources entirely. The root agent receives all its incident
resources entirely. The resulting allocation is integral and feasible. Denote the resulting binary allocation
by \(\widehat{x}^{\mathrm F}(u)\), and define
\begin{align}
\underline{m}_{N,K}(u)
:=
\min_{1\leq i\leq N}
\sum_{j=1}^{K}u_{i,j}\widehat{x}^{\mathrm F}_{i,j}(u).
\end{align}

The rounded allocation assigns every resource to exactly one agent and is
therefore feasible for the original indivisible-allocation problem. Since
\(m_{N,K}(u)\) is the maximum minimum utility over all feasible integral
allocations, the value attained by this particular allocation satisfies

\begin{align}
\underline{m}_{N,K}(u)\leq m_{N,K}(u).
\end{align}

For the random utility matrix, we write

\begin{align}
\underline{M}_{N,K}
:=
\underline{m}_{N,K}\!\left(U_{(N,K)}\right)
\end{align}

Recall that
\begin{align}
    \mu_U^N
    :=
    \mathbb{E}\left[U_{[N:N]}\right],
\end{align}
where \(U_{[N:N]}\) denotes the maximum of \(N\) independent utility values.
For i.i.d.\ utilities, Theorem~\ref{thm:fixed-agent} gives
\begin{align}
    M_{N,K}
    \sim
    \frac{K}{N}\mu_U^N,
    \qquad K\to\infty.
    \label{eq:sim-fixed-agent-benchmark}
\end{align}

Figure~\ref{fig:simulations-fixed-agents} displays the \(5\)th percentile,
median, and \(95\)th percentile of both the empirical lower and upper bounds,
together with the theoretical value \(\frac{K}{N}\mu_U^N\), for \(N=2\) and
\(N=32\).
Figure~\ref{fig:simulations-fixed-normalized} additionally compares the seven
utility distributions through the two normalized medians
\begin{equation}
    \frac{\operatorname{median}(\underline{M}_{N,K})}
    {\left(\frac{K}{N}\right)\mu_U^{N}},
    \qquad
    \frac{\operatorname{median}(\overline{M}_{N,K})}
    {\left(\frac{K}{N}\right)\mu_U^{N}},
    \qquad N\in\{2,32\}.
    \label{eq:sim-fixed-normalized-bound-medians}
\end{equation}
The horizontal reference line at \(1\) represents the limit predicted by
Theorem~\ref{thm:fixed-agent}. Because the exact max-min value lies between
the two bounds, the approach of both normalized medians toward \(1\) provides
numerical support for the theoretical characterization, while their
separation illustrates the remaining finite-size approximation gap. The
normalized curves also allow the convergence rates of the Weibull, uniform,
Nakagami, Rician, Pareto type II, and exponential models to be compared on a common
scale. For the Pareto type II case, the comparison with
Theorem~\ref{thm:fixed-agent} requires a finite variance, thus $\alpha=3$ is chosen.

\begin{figure*}[htbp]
    \centering

    \begin{subfigure}[t]{0.32\textwidth}
        \centering
        \includegraphics[width=\linewidth]
        {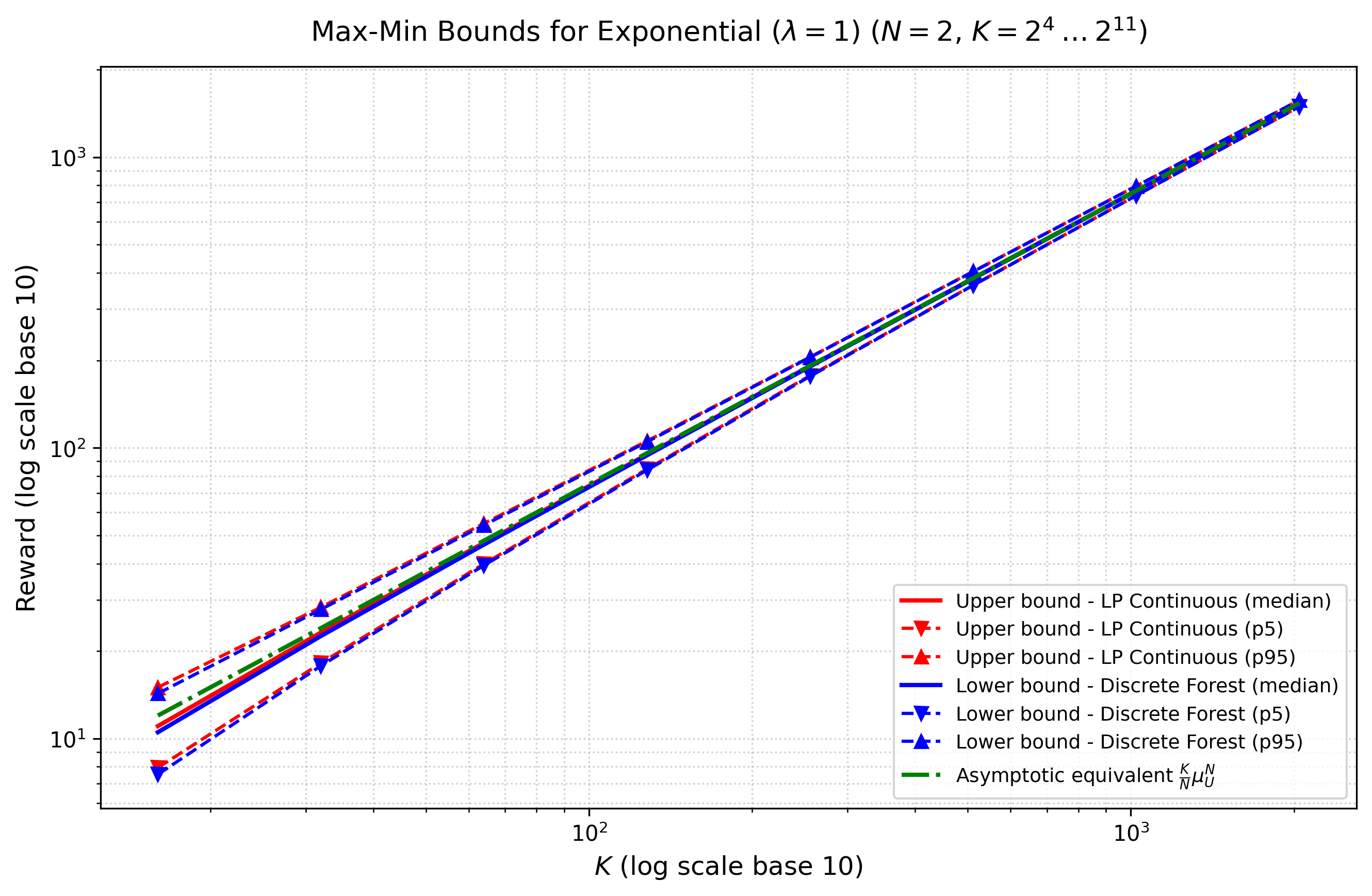}
        \caption{Exponential distribution, \(N=2\).}
        \label{fig:fixed-exponential}
    \end{subfigure}
    \hfill
    \begin{subfigure}[t]{0.32\textwidth}
        \centering
        \includegraphics[width=\linewidth]
        {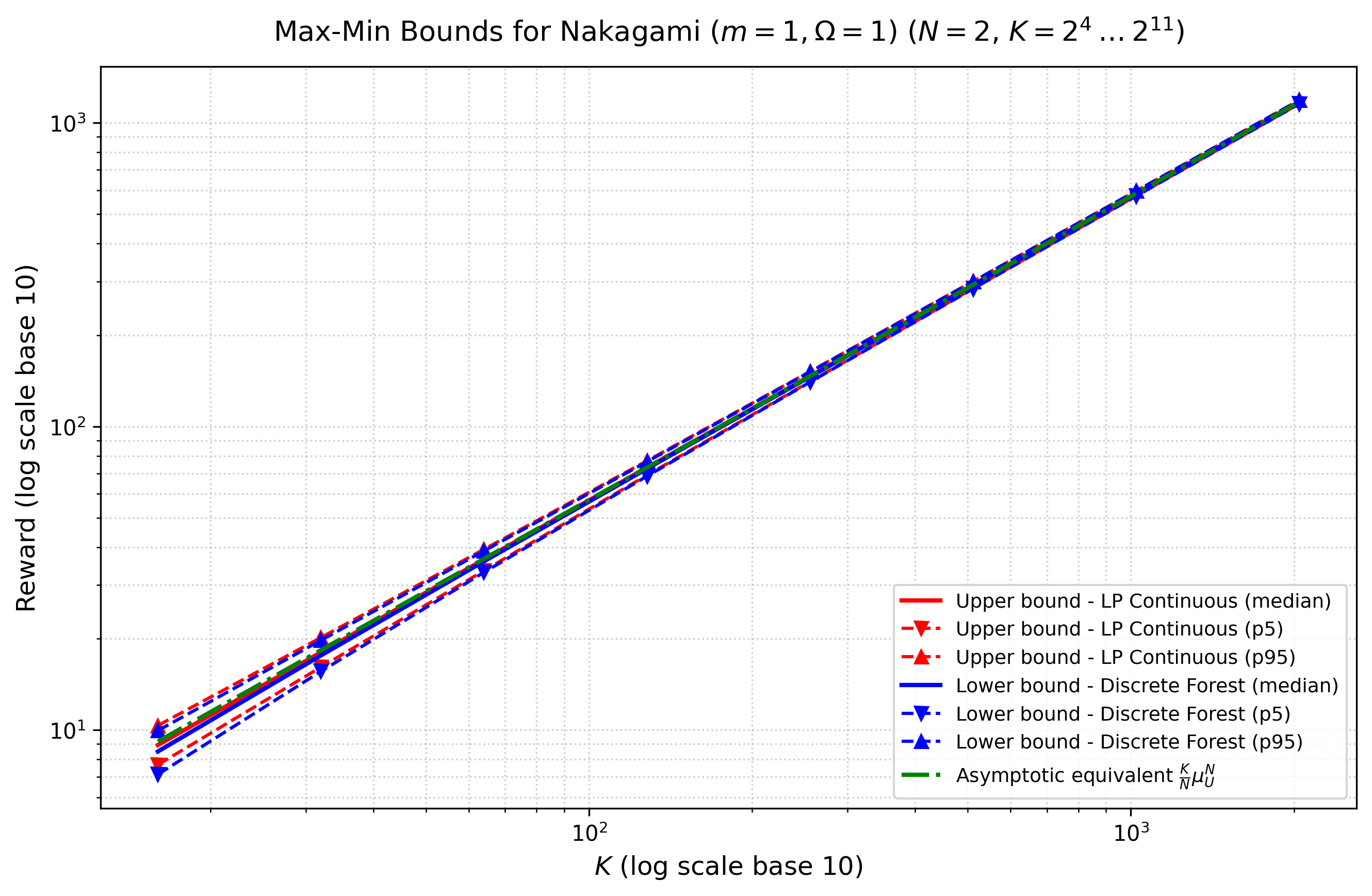}
        \caption{Nakagami distribution, \(N=2\).}
        \label{fig:fixed-nakagami}
    \end{subfigure}
    \hfill
    \begin{subfigure}[t]{0.32\textwidth}
        \centering
        \includegraphics[width=\linewidth]
        {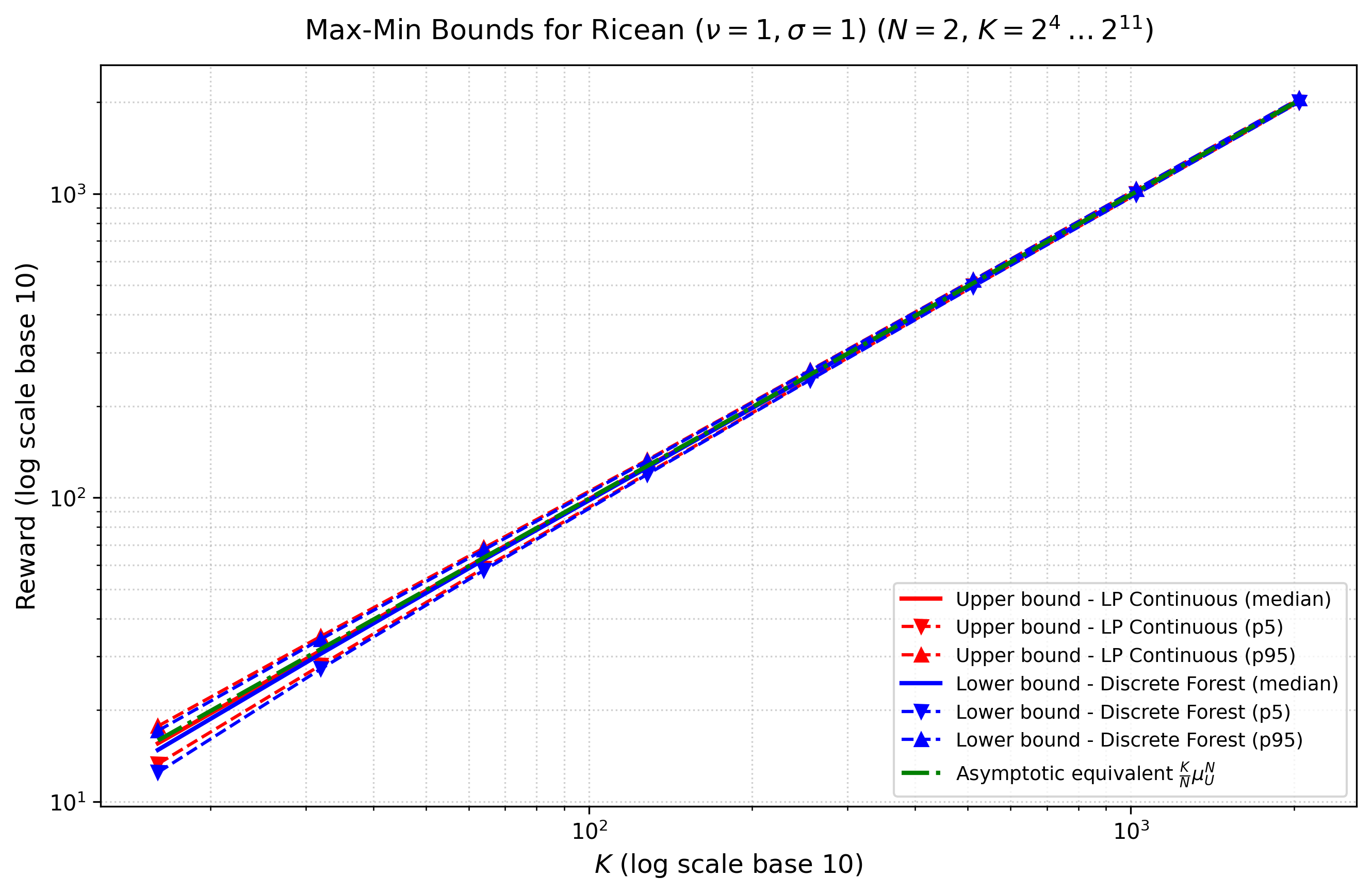}
        \caption{Rician distribution, \(N=2\).}
        \label{fig:fixed-rice}
    \end{subfigure}

    \medskip

    \begin{subfigure}[t]{0.32\textwidth}
        \centering
        \includegraphics[width=\linewidth]
        {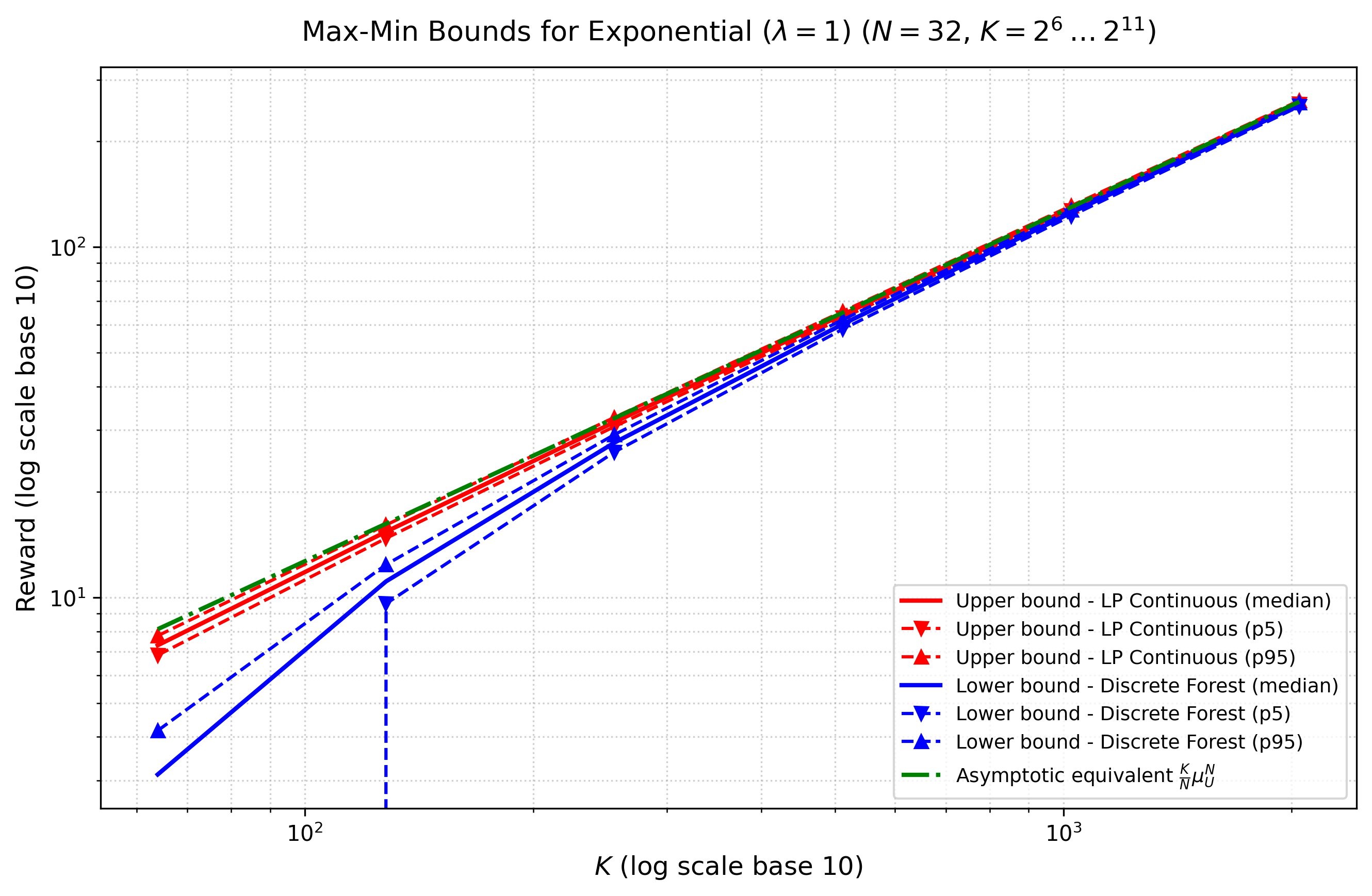}
        \caption{Exponential distribution, \(N=32\).}
        \label{fig:fixed-exponential-32}
    \end{subfigure}
    \hfill
    \begin{subfigure}[t]{0.32\textwidth}
        \centering
        \includegraphics[width=\linewidth]
        {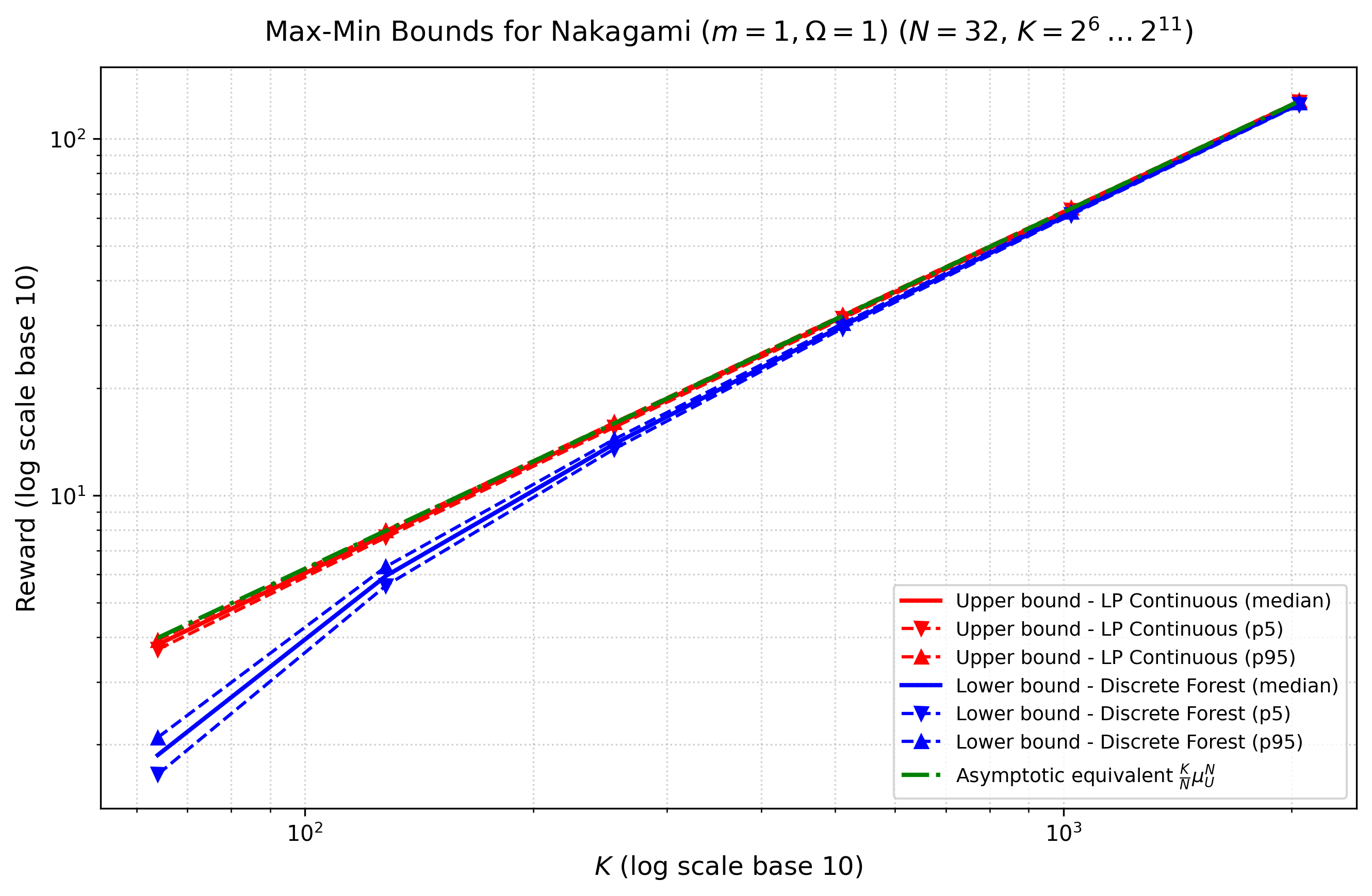}
        \caption{Nakagami distribution, \(N=32\).}
        \label{fig:fixed-nakagami-32}
    \end{subfigure}
    \hfill
    \begin{subfigure}[t]{0.32\textwidth}
        \centering
        \includegraphics[width=\linewidth]
        {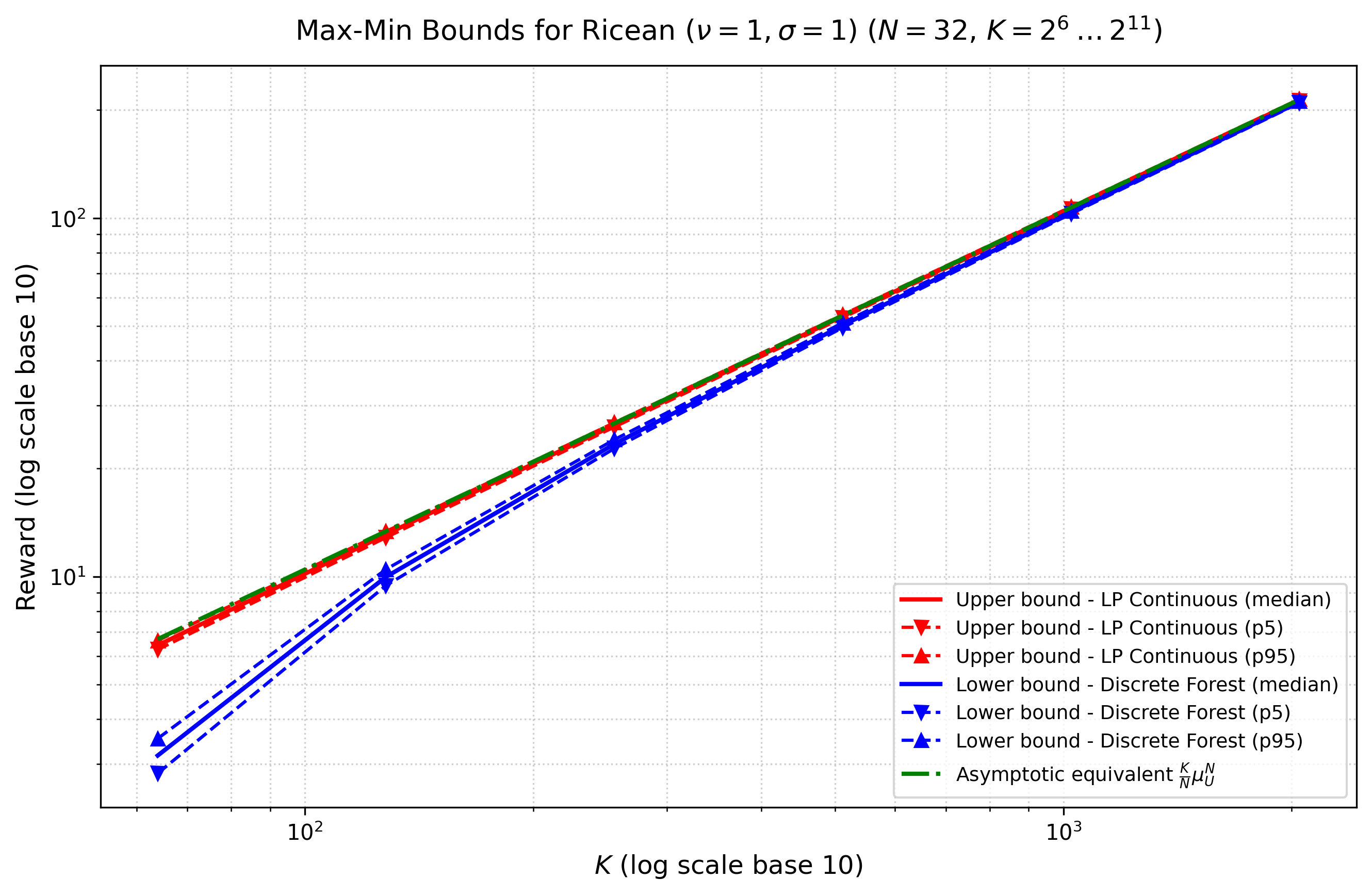}
        \caption{Rician distribution, \(N=32\).}
        \label{fig:fixed-rice-32}
    \end{subfigure}

    \caption{Empirical bounds in the fixed-agent regime with \(N=2\) and
    \(N=32\). The exponential, Nakagami, and Rician parameters are
    \(\lambda=1\), \((m,\Omega)=(1,1)\), and \((\nu,\sigma)=(1,1)\),
    respectively. For \(N=2\), \(K\) ranges from \(2^4\) to \(2^{11}\);
    for \(N=32\), it ranges from \(2^6\) to \(2^{11}\), in powers of two.
    For both the forest-rounded lower bound and the linear-programming upper
    bound, the \(5\)th percentile, median, and \(95\)th percentile over
    \(1{,}000\) Monte Carlo trials are presented as functions of \(K\).
    The theoretical curve is \(\frac{K}{N}\mu_U^N\), as established in
    Theorem~\ref{thm:fixed-agent}.}
    \label{fig:simulations-fixed-agents}
\end{figure*}

\begin{figure*}[htbp]
    \centering

    \begin{subfigure}[t]{0.48\textwidth}
        \centering
 \includegraphics[width=\linewidth]
        {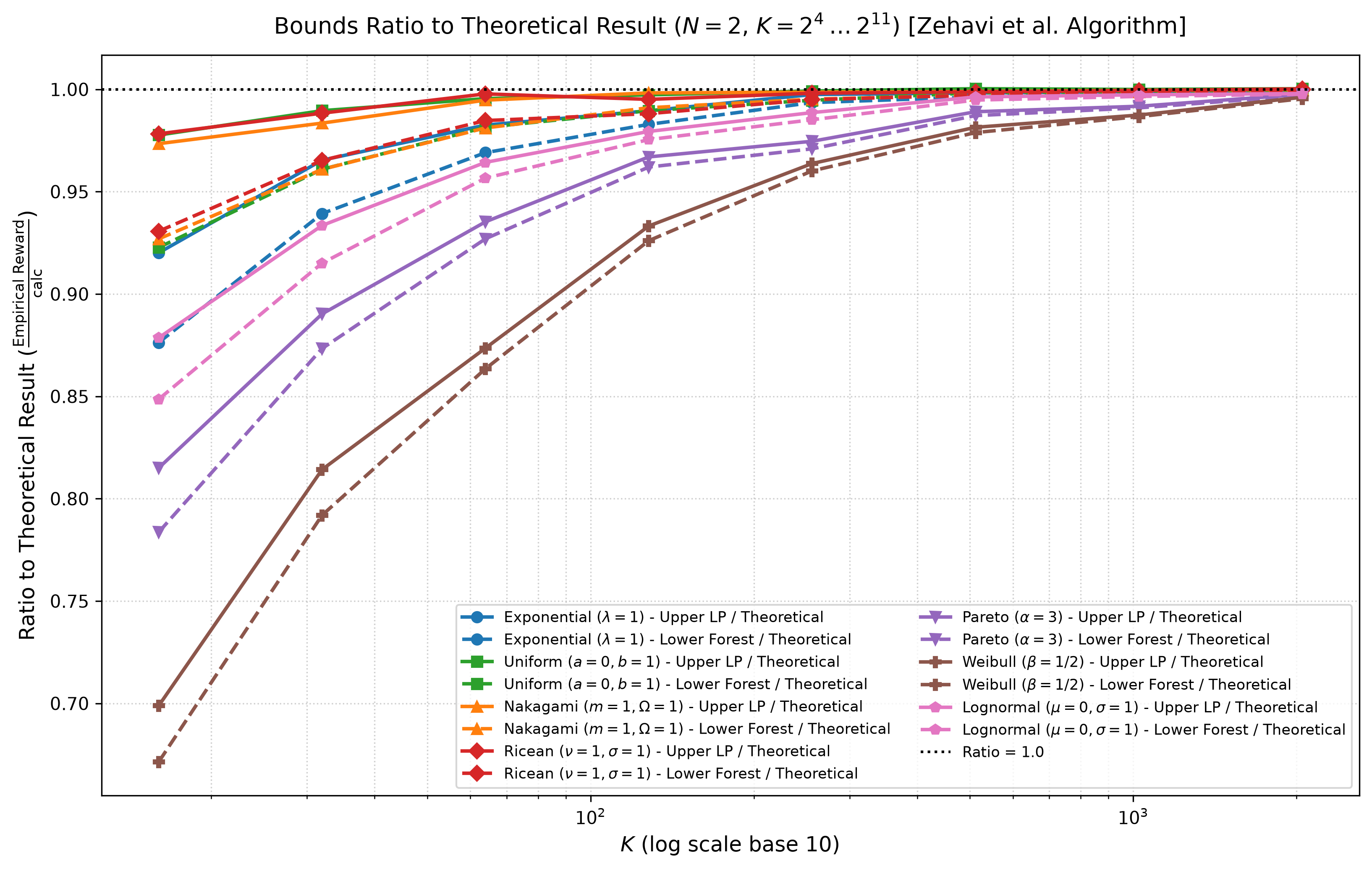}
        \caption{\(N=2\), with \(K=2^4,\ldots,2^{11}\).}
        \label{fig:fixed-normalized-n2}
    \end{subfigure}
    \hfill
    \begin{subfigure}[t]{0.48\textwidth}
        \centering
\includegraphics[width=\linewidth]
        {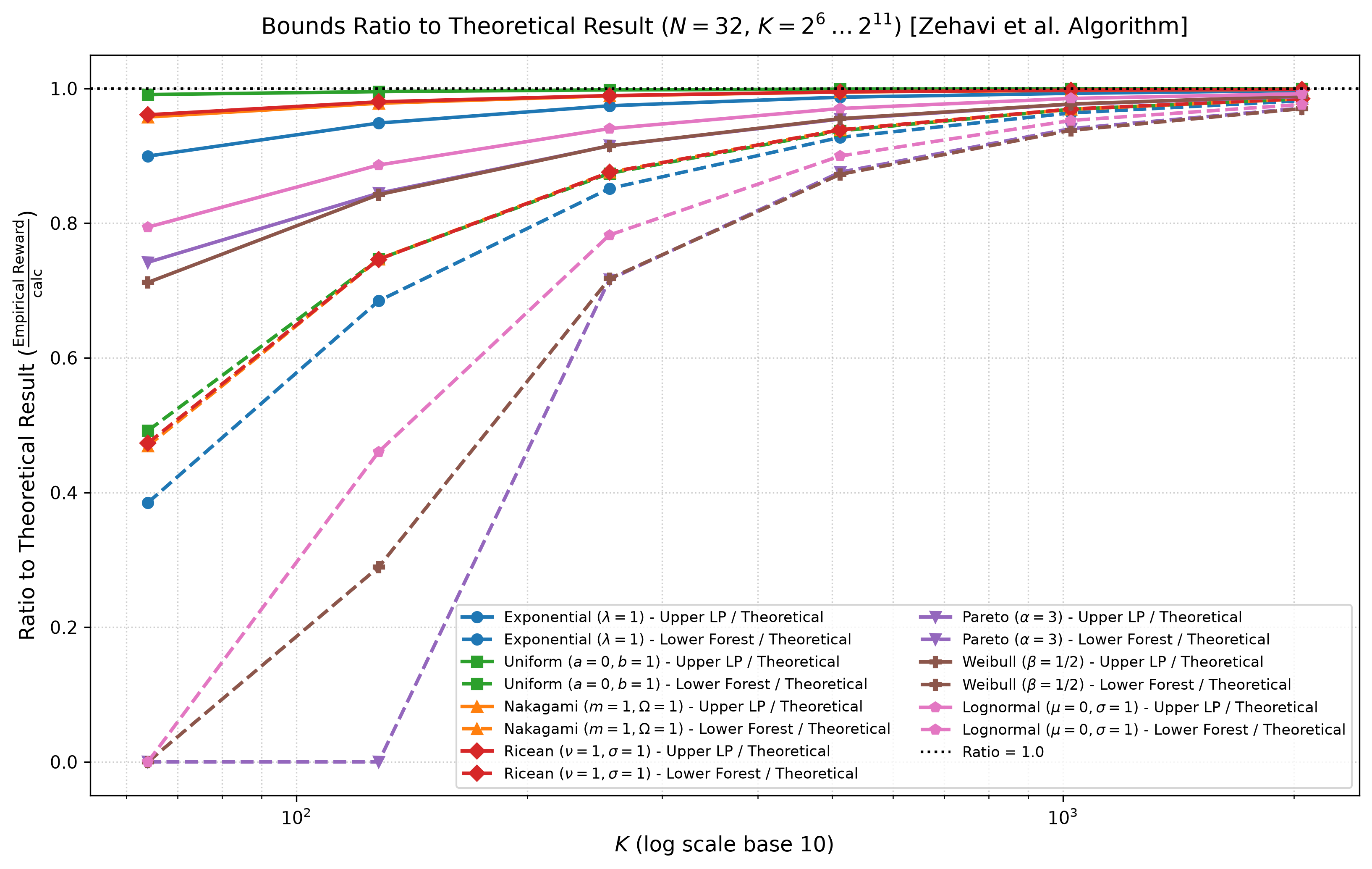}
        \caption{\(N=32\), with \(K=2^6,\ldots,2^{11}\).}
        \label{fig:fixed-normalized-n32}
    \end{subfigure}

    \caption{Normalized empirical lower and upper-bound medians in the
    fixed-agent regime for \(N=2\) and \(N=32\). For each of the seven utility
    distributions, both panels show
    \(\operatorname{median}(\underline{M}_{N,K})/
    (\frac{K}{N}\mu_U^{N})\) and
    \(\operatorname{median}(\overline{M}_{N,K})/
    (\frac{K}{N}\mu_U^{N})\), computed over \(1{,}000\) Monte Carlo trials.
    The horizontal reference line is \(1\).}
    \label{fig:simulations-fixed-normalized}
\end{figure*}

\section{Conclusion}
\label{sec:conclusion}

We developed a tail-based asymptotic theory for max-min fair allocation under i.i.d.\ random utilities. In contrast to distribution-specific analyses, our results identify explicit upper-tail conditions that determine the asymptotic max-min value. In the balanced regime $K=N$, the perfect-matching threshold yields a quantile characterization valid for the broad class of distributions with regularly varying upper-tail quantiles, with stronger additive convergence for admissible Weibull-type tails. In the proportional-growth regime $K=LN$, an $L$-matching lower bound and a sum-welfare upper bound yield

$$
M_{N,LN}\sim L F^{-1}\!\left(1-\frac{\ln N}{N}\right)
$$

for Weibull-type utilities and distributions with a finite positive upper endpoint.

These characterizations also establish that, in the proportional-growth regime under the stated assumptions, every max-min fair allocation has asymptotically optimal aggregate welfare, and hence its price of fairness converges to zero. In the complementary fixed-agent regime, $N$ fixed and $K\to\infty$, the opportunistic allocation is asymptotically max-min fair optimal and

$$
\frac{NM_{N,K}}{S_{N,K}}\xrightarrow{P}1,
$$

again implying a vanishing price of fairness.

Thus, broad stochastic utility models admit strong max-min fairness with asymptotically negligible efficiency loss. Natural extensions include sharper second-order asymptotics, heavy-tailed proportional-growth regimes, and dependent or heterogeneous utility models.
\section*{Acknowledgement}
OpenAI ChatGPT5.6 was used during the preparation of this manuscript.
The system was used to assist with drafting and revising portions of
the text, improving organization and presentation, checking and clarifying mathematical arguments. Antigravity was used to develop and refine the code used for the numerical experiments.  The authors reviewed and verified all mathematical
arguments, numerical results, references, and final manuscript text, and take full responsibility for the content of the paper.

\begin{appendices}
\section{Proof of Lemma~\ref{Lemma: total convergence order statistic complete}}
\label{app: special case}
Let $\gz_1^N=\frac{\ln N+\ln \ln \ln N}{N},
\gz_2^N=\frac{\ln N - \ln \ln \ln N}{N}$.
In this appendix, we will prove the following lemma:
\admissible*
\begin{proof}
Suppose first that $F$ has a finite positive upper endpoint
\begin{equation}
x_F := \sup\{x \in \mathbb{R} : F(x) < 1\} \in (0,\infty).
\end{equation}

Since
\begin{equation}
\gz_2^N\longrightarrow 0
\qquad \text{and} \qquad
\gz_1^N \longrightarrow 0,
\end{equation}
the definition of the upper endpoint gives
\begin{equation}
F^{-1}\left(1 - \gz_2^N\right) \longrightarrow x_F
\end{equation}
and
\begin{equation}
F^{-1}\left(1 - \gz_1^N\right) \longrightarrow x_F.
\end{equation}

Consequently,
\begin{equation}
F^{-1}\left(1-\gz_2^N\right)
-F^{-1}\left(1-\gz_1^N\right)
\longrightarrow0.
\end{equation}

This proves the result in the bounded-support case.
Now suppose that $F$ has an admissible Weibull-type tail. Define
\begin{equation}
V_F(y):=F^{-1}(1-e^{-y}).
\end{equation}
Assume that, for some $\theta\in[0,2)$,
\begin{equation}
V_F(y)=y^\theta L(y),
\end{equation}
where $L$ is slowly varying, $V_F$ is eventually continuously
differentiable, and
\begin{equation}
\frac{yV_F'(y)}{V_F(y)}\longrightarrow\theta.
\end{equation}
Let \begin{equation} x_N:=\ln N \qquad\text{and}\qquad h_N:=\ln\ln\ln N. \end{equation} $\gz_1^N,\gz_2^N\in(0,1)$. Define \begin{equation} a_N := \ln\left(\frac{1}{\gz_2^N}\right) = \ln N-\ln\left(\ln N-\ln\ln\ln N\right) \end{equation} and \begin{equation} b_N := \ln\left(\frac{1}{\gz_1^N}\right) = \ln N-\ln\left(\ln N+\ln\ln\ln N\right). \end{equation} By the definition of $V_F$, \begin{equation} F^{-1}\left(1-\gz_2^N\right)=V_F(a_N) \end{equation} and \begin{equation} F^{-1}\left(1-\gz_1^N\right)=V_F(b_N). \end{equation} Hence it is enough to prove that \begin{equation} V_F(a_N)-V_F(b_N)\longrightarrow0. \end{equation} First, since \begin{equation} \ln\left(x_N\pm h_N\right)=o(x_N), \end{equation} we have \begin{equation} \frac{a_N}{x_N}\longrightarrow1 \qquad\text{and}\qquad \frac{b_N}{x_N}\longrightarrow1. \end{equation} Moreover, \begin{align} a_N-b_N &= \ln\left( \frac{x_N+h_N}{x_N-h_N} \right) \notag\\ &= \ln\left( \frac{1+h_N/x_N}{1-h_N/x_N} \right). \label{eq:quantile-argument-difference} \end{align} Since $h_N/x_N\to0$ and \begin{equation} \ln\left(\frac{1+z}{1-z}\right) = 2z+o(z) \qquad\text{as }z\to0, \end{equation} it follows that \begin{equation} a_N-b_N = \frac{2h_N}{x_N}\bigl(1+o(1)\bigr). \end{equation} In particular, \begin{equation} a_N-b_N = O\left(\frac{h_N}{x_N}\right). \end{equation} Because $V_F$ is eventually continuously differentiable, the mean-value theorem implies that, for all sufficiently large $N$, there exists \begin{equation} \xi_N\in(b_N,a_N) \end{equation} such that \begin{equation} V_F(a_N)-V_F(b_N) = V_F'(\xi_N)(a_N-b_N). \end{equation} Since $a_N/x_N\to1$ and $b_N/x_N\to1$, we also have \begin{equation} \frac{\xi_N}{x_N}\longrightarrow1. \end{equation} The assumption \begin{equation} \frac{yV_F'(y)}{V_F(y)}\longrightarrow\theta \end{equation} implies that this ratio is bounded for all sufficiently large $y$. Thus, there exists a constant $C>0$ such that \begin{equation} \left|V_F'(\xi_N)\right| \le C\frac{V_F(\xi_N)}{\xi_N} \end{equation} for all sufficiently large $N$. Since $V_F\in\mathrm{RV}_\theta$ and $\xi_N/x_N\to1$, the uniform convergence theorem for regularly varying functions gives \begin{equation} \frac{V_F(\xi_N)}{V_F(x_N)} \longrightarrow1. \end{equation} Also, $\xi_N/x_N\to1$. Consequently, \begin{equation} \left|V_F'(\xi_N)\right| = O\left(\frac{V_F(x_N)}{x_N}\right). \end{equation} Combining this estimate with \begin{equation} a_N-b_N = O\left(\frac{h_N}{x_N}\right), \end{equation} we obtain \begin{align*} 0 &\le V_F(a_N)-V_F(b_N)\\ &= O\left( \frac{h_NV_F(x_N)}{x_N^2} \right). \end{align*} Using the representation \begin{equation} V_F(x_N)=x_N^\theta L(x_N), \end{equation} this becomes \begin{equation} 0 \le V_F(a_N)-V_F(b_N) = O\left( h_Nx_N^{\theta-2}L(x_N) \right). \end{equation} Finally, because $x_N=\ln N$, \begin{equation} h_N = \ln\ln\ln N = \ln\ln x_N. \end{equation} The function \begin{equation} \widetilde L(x):=(\ln\ln x)L(x) \end{equation} is slowly varying. Since $\theta<2$, the standard growth property of slowly varying functions gives \begin{equation} x^{\theta-2}\widetilde L(x) \longrightarrow0. \end{equation} Therefore, \begin{equation} h_Nx_N^{\theta-2}L(x_N) \longrightarrow0. \end{equation} It follows that \begin{equation} V_F(a_N)-V_F(b_N)\longrightarrow0, \end{equation} and hence \begin{equation} F^{-1}\left(1-\gz_2^N\right) - F^{-1}\left(1-\gz_1^N\right) \longrightarrow0. \end{equation} \end{proof}

\section{Proof of Lemma~\ref{lemma:regular-tail-quantile-threshold}}
\label{app:regular-tail-quantiles}

\regular*

\begin{proof}
Suppose first that $F$ has a Weibull-type tail. Then its log-tail
quantile $V_F$ is regularly varying at infinity with some finite index
$\theta\geq0$. Recall that
\begin{equation}
Q_F(t)=V_F(\ln t).
\end{equation}
Fix $c>0$ and set $y:=\ln t$. Then
\begin{equation}
\frac{Q_F(ct)}{Q_F(t)}
=
\frac{V_F(y+\ln c)}{V_F(y)}
=
\frac{V_F(\lambda_y y)}{V_F(y)},
\qquad
\lambda_y:=1+\frac{\ln c}{y}\longrightarrow1.
\end{equation}
Because $V_F$ is regularly varying, the uniform convergence theorem
implies that the last ratio converges to one. Hence $Q_F$ is slowly
varying, or equivalently regularly varying with coefficient $\rho=0$.
 Let $h_N:=\ln\ln\ln N$ and define
\begin{equation}
r_N:=\frac{N}{\ln N+h_N},
\qquad
s_N:=\frac{N}{\ln N-h_N}.
\end{equation}
For all sufficiently large $N$, both sequences tend to infinity and
\begin{equation}
\lambda_N:=\frac{r_N}{s_N}
=\frac{\ln N-h_N}{\ln N+h_N}
\longrightarrow1.
\end{equation}

Suppose that $F$ has a regularly varying upper-tail quantile with
coefficient $\rho\geq0$. By the definition of $Q_F$,
\begin{equation}
F^{-1}(1-\gz_1^N)=Q_F(r_N),
\qquad
F^{-1}(1-\gz_2^N)=Q_F(s_N).
\end{equation}
Since $r_N=\lambda_Ns_N$, the uniform convergence theorem for regularly
varying functions gives
\begin{equation}
\frac{F^{-1}(1-\gz_1^N)}{F^{-1}(1-\gz_2^N)}
=\frac{Q_F(\lambda_Ns_N)}{Q_F(s_N)}
\longrightarrow1.
\end{equation}
This proves
\eqref{eq:regular-tail-quantile-matching-ratio} for the unbounded case.

Finally, suppose that $F$ has a finite positive upper endpoint
\begin{equation}
x_F:=\sup\{x\in\mathbb{R}:F(x)<1\}\in(0,\infty).
\end{equation}
Since $\gz_1^N\to0$ and $\gz_2^N\to0$,
\begin{equation}
F^{-1}(1-\gz_1^N)\longrightarrow x_F,
\qquad
F^{-1}(1-\gz_2^N)\longrightarrow x_F.
\end{equation}
Their ratio therefore converges to one, completing the proof.
\end{proof}
\section{Proof of Lemma~\ref{theorem:main-first-model}}
\label{app:Proportional Growth: Lower Bound}

\lowerbound*

\begin{proof}

To prove this, we adapt the methodology used for the single-resource case in Section~\ref{sec: Single Resource Per agent}. First, Definition~\ref{def:L-matching} introduces the concept of an L-matching, which generalizes perfect matching to many resources per agent. Next, in lemma~\ref{lemma:erdos_reyni_equivalent}, we present a result analogous to the \ER theorem~\cite{erdos1966random}, extended to the context of L-matchings. Finally, applying the framework from the single-resource setting, we derive a lower bound for the asymptotic max-min value.
\begin{definition}
    Let $G(V,W)$ be a bipartite graph, and let $L$ be an integer. $G(V,W)$ has an L-matching if every left vertex can select $L$ neighbors, with no two left vertices sharing the same right vertex.
    \label{def:L-matching}
\end{definition}
Note that an allocation of $L$ resources to each agent is equivalent to an L-matching in $G(V,W)$. 
For $L=1$, an L-matching is a perfect matching.
\begin{claim} 
Let $\mathbf u$ be a fixed $N\times LN$ utility matrix. If there is an L-matching in the bipartite graph
\begin{equation}
G_{\mathbf u}(\gt)
=
\bigl(A,B,E_{\mathbf u}(\gt)\bigr),
\qquad
E_{\mathbf u}(\gt)
=
\left\{(i,j):u_{i,j}\geq\gt\right\},
\end{equation}
then
\begin{equation}
L\gt\leq m_{N,LN}(\mathbf u).
\end{equation}
\label{claim:K=LN}
\end{claim}
Indeed, an L-matching defines an allocation under which every agent
receives $L$ resources, each with utility at least $\gt$. Thus every
agent receives total utility at least $L\gt$, and hence
$m_{N,LN}(\mathbf u)\geq L\gt$.
\begin{lemma}
\label{lemma:erdos_reyni_equivalent}
Let \(L\geq 1\) be a fixed integer and let \(K=LN\). Let
\(\omega=\omega(N)\) and
\begin{equation}
    p=\frac{\ln N+\omega}{N},
\end{equation}
where \(p\in[0,1]\) for all sufficiently large \(N\). Let \(A_N\)
denote the event that the random bipartite graph
\(G_{N,LN,p}\) contains an \(L\)-matching. Then
\begin{equation}
\lim_{N\to\infty} \mathbb{P}(A_N)
=
\begin{cases}
0, & \text{if } \omega\to-\infty,\\
1, & \text{if } \omega\to\infty.
\end{cases}
\end{equation}
\end{lemma}

\begin{proof}
Let the two vertex classes of \(G_{N,LN,p}\) have sizes \(N\)
and \(LN\), respectively, and suppose first that
\(\omega\to-\infty\).

Let \(Z_N\) denote the number of isolated vertices in the vertex
class of size \(LN\). A fixed vertex in this class is isolated with
probability
\begin{equation}
    \alpha_N=(1-p)^N.
\end{equation}
Since the sets of edges incident to distinct vertices in this class
are disjoint, the corresponding isolation events are independent.
Consequently,
\begin{equation}
    Z_N\sim\operatorname{Bin}(LN,\alpha_N).
\end{equation}

An \(L\)-matching contains \(LN\) edges with distinct endpoints in
the vertex class of size \(LN\). Since that class contains exactly
\(LN\) vertices, the existence of an \(L\)-matching requires every
such vertex to be nonisolated. Therefore,
\begin{equation}
    A_N\subseteq\{Z_N=0\}.
\end{equation}

Because \(\omega\to-\infty\), we have
\(p\leq \frac{\ln N}{N}\) for all sufficiently large \(N\), and hence
\(Np^2=o(1)\). It follows that
\begin{align*}
    \ln\left(LN\alpha_N\right)
    &=\ln L+\ln N+N\ln(1-p)\\
    &=\ln L+\ln N-Np+O(Np^2)\\
    &=\ln L-\omega+o(1)
      \longrightarrow\infty.
\end{align*}
Thus \(LN\alpha_N\to\infty\), and
\begin{equation}
\begin{aligned}
    \mathbb{P}(A_N)
    &\leq \mathbb{P}(Z_N=0)\\
    &=(1-\alpha_N)^{LN}\\
    &\leq e^{-LN\alpha_N}
      \longrightarrow 0.
\end{aligned}
\end{equation}

Now suppose that \(\omega\to\infty\). Partition the vertex class of
size \(LN\) into \(L\) disjoint sets
\begin{equation}
    B_1,\ldots,B_L,
    \qquad |B_\ell|=N.
\end{equation}
For each \(\ell\in\{1,\ldots,L\}\), let \(A_N^{(\ell)}\) be the
event that the bipartite graph induced by the vertex class of size
\(N\) and \(B_\ell\) contains a perfect matching. Each such induced
graph is distributed as \(G_{N,N,p}\).

By the Erd\H{o}s-R\'enyi perfect-matching threshold theorem
\cite{erdos1966random,erdos1968random}
(see also \cite[Theorem~6.1]{frieze2016introduction}),
\begin{equation}
    \mathbb{P}\left(A_N^{(\ell)}\right)\longrightarrow 1
    \qquad\text{for every }\ell\in\{1,\ldots,L\}.
\end{equation}
Since \(L\) is fixed, the union bound gives
\begin{equation}
\begin{aligned}
    \mathbb{P}\left(\bigcap_{\ell=1}^{L}A_N^{(\ell)}\right)
    &\geq
    1-\sum_{\ell=1}^{L}\mathbb{P}\left(\left(A_N^{(\ell)}\right)^c\right)\\
    &\longrightarrow 1.
\end{aligned}
\end{equation}
On the event \(\bigcap_{\ell=1}^{L}A_N^{(\ell)}\), the union of the
\(L\) perfect matchings assigns one distinct resource from each
block \(B_\ell\) to every agent. Their union is therefore an
\(L\)-matching in \(G_{N,LN,p}\). Hence
\begin{equation}
    \mathbb{P}(A_N)
    \geq
    \mathbb{P}\left(\bigcap_{\ell=1}^{L}A_N^{(\ell)}\right)
    \longrightarrow 1.
\end{equation}
This completes the proof.
\end{proof}

Let
\begin{equation}
\gz_1^N
:=
\frac{\ln N+\ln\ln\ln N}{N},
\qquad
\gt_N
:=
F^{-1}\left(1-\gz_1^N\right).
\end{equation}
The random threshold graph
$G_{\mU_{(N,LN)}}(\gt_N)$ has edge probability $\gz_1^N$.
Therefore, by Lemma~\ref{lemma:erdos_reyni_equivalent}, it contains
an L-matching with probability tending to one. Applying
Claim~\ref{claim:K=LN} realization-wise shows that, with probability
tending to one,
\begin{align}
\label{eq:main-theorem-proportional}
L F^{-1}\left(1-\gz_1^N\right)\leq M_{N,LN}.
\end{align}

Applying
Lemma~\ref{Lemma: total convergence order statistic complete}  to \eqref{eq:main-theorem-proportional} proves part~(b) of
Lemma~\ref{theorem:main-first-model}. Similarly, applying
Lemma~\ref{lemma:regular-tail-quantile-threshold} proves
part~(a).
    
\end{proof}
Computing $F^{-1}\left(1-\frac{\ln N}{N}\right)$ requires knowledge of the quantile function of $F$. In many cases, this can either be computed (e.g., for a Rayleigh fading channel) or well approximated. 
\section{Proof of Lemma~\ref{lemma: Asymptotic equivalence of ratio quantiles}}
\label{app:first-model}

\weibull*

\begin{proof}
Define the log-tail quantile function by
\begin{equation}
V_{F}(y):=F^{-1}\left(1-e^{-y}\right),
\qquad y>0.
\end{equation}
Suppose first that $F$ has a finite positive upper endpoint
\begin{equation}
x_F:=\sup\{x\in\mathbb{R}:F(x)<1\}\in(0,\infty).
\end{equation}
Since
\begin{equation}
\frac1N\longrightarrow0
\qquad\text{and}\qquad
\gz_1^N\longrightarrow0,
\end{equation}
the definition of the upper endpoint gives
\begin{equation}
F^{-1}\left(1-\frac1N\right)\longrightarrow x_F
\end{equation}
and
\begin{equation}
F^{-1}\left(1-\gz_1^N\right)\longrightarrow x_F.
\end{equation}
Because $x_F>0$, it follows that
\begin{equation}
\frac{
F^{-1}\left(1-\frac1N\right)
}{
F^{-1}\left(1-\gz_1^N\right)
}
\longrightarrow
\frac{x_F}{x_F}
=1.
\end{equation}
This proves the result in the bounded-support case.

We now suppose that $F$ has an unbounded upper endpoint and that
$V_{F}$ is regularly varying at infinity with some
finite index $\theta\geq 0$; that is,
\begin{equation}
\lim_{y\to\infty}
\frac{V_{F}(cy)}{V_{F}(y)}
=
c^\theta
\qquad\text{for every fixed }c>0.
\end{equation}

For all sufficiently large $N$, we have
\begin{equation}
0<\gz_1^N<1.
\end{equation}
By the definition of $V_{F}$,
\begin{equation}
F^{-1}\left(1-\frac{1}{N}\right)
=
V_{F}(\ln N),
\end{equation}
whereas
\begin{equation}
F^{-1}\left(1-\gz_1^N\right)
=
V_{F}\left(\ln\frac{1}{\gz_1^N}\right)
=
V_{F}\left(
\ln\frac{N}{\ln N+\ln\ln\ln N}
\right).
\end{equation}

Define
\begin{equation}
a_N:=\ln N
\end{equation}
and
\begin{equation}
b_N
:=
\ln\frac{N}{\ln N+\ln\ln\ln N}
=
\ln N-\ln\bigl(\ln N+\ln\ln\ln N\bigr).
\end{equation}
Since
\begin{equation}
\frac{
\ln\bigl(\ln N+\ln\ln\ln N\bigr)
}{
\ln N
}
\longrightarrow 0,
\end{equation}
we have
\begin{equation}
b_N\longrightarrow\infty
\end{equation}
and
\begin{equation}
\frac{a_N}{b_N}
=
\frac{\ln N}{
\ln N-\ln\bigl(\ln N+\ln\ln\ln N\bigr)
}
\longrightarrow 1.
\end{equation}

Set
\begin{equation}
\lambda_N:=\frac{a_N}{b_N}.
\end{equation}
Then $\lambda_N\to1$ and $a_N=\lambda_N b_N$. Therefore,
\begin{equation}
\frac{
F^{-1}\left(1-\frac{1}{N}\right)
}{
F^{-1}\left(1-\gz_1^N\right)
}
=
\frac{V_{F}(a_N)}{V_{F}(b_N)}
=
\frac{V_{F}(\lambda_Nb_N)}{V_{F}(b_N)}.
\end{equation}

Since $V_{F}\in\operatorname{RV}_{\theta}$, the uniform convergence
theorem for regularly varying functions yields
\begin{equation}
\frac{V_{F}(\lambda_Nb_N)}{V_{F}(b_N)}
\longrightarrow 1,
\end{equation}
because $b_N\to\infty$ and $\lambda_N\to1$. Hence,
\begin{equation}
\frac{
F^{-1}\left(1-\frac{1}{N}\right)
}{
F^{-1}\left(1-\gz_1^N\right)
}
\longrightarrow 1.
\end{equation}
\end{proof}

\section{Proportional Growth: Upper Bound}
\label{app:Proportional Growth: Upper Bound}
To gain a better understanding of the behavior of the optimal sum welfare value, we first need to understand the nature of the bounds of the max-min value.
The following lemma shows that the upper bound for the max-min value,
$S_{N,LN}/N$, behaves asymptotically as
$L F^{-1}\left(1-\frac{1}{N}\right)$.
\upperbound*
\begin{proof}
To prove this, we first analyze the asymptotic behavior of the
normalized optimal sum-welfare value.
\begin{lemma}[Averaging independent maxima]
\label{lemma:averaging-independent-maxima}
Let $F$ be a continuous distribution function supported on
$[0,\infty)$. Assume that either:
\begin{enumerate}
    \item $F$ has a finite positive upper endpoint; or
    \item $F$ has an unbounded upper endpoint and
    \begin{equation}
        V_{F}(y):=F^{-1}(1-e^{-y})
        \in\operatorname{RV}_{\theta}
    \end{equation}
    for some $\theta\geq0$.
\end{enumerate}

Define
\begin{equation}
b_N:=F^{-1}\left(1-\frac{1}{N}\right),
\end{equation}
let $U_{[N:N]}$ denote the maximum of $N$ independent random
variables with distribution function $F$, and let
\begin{equation}
U_{[N:N]}^{(1)},\ldots,U_{[N:N]}^{(LN)}
\end{equation}
be independent copies of $U_{[N:N]}$, where $L\geq1$ is a fixed
integer. Then
\begin{equation}
\frac{1}{LNb_N}
\sum_{j=1}^{LN}U_{[N:N]}^{(j)}
\xrightarrow{L^2}1.
\end{equation}
In particular, the convergence also holds in probability.
\end{lemma}

\begin{proof}
Set
\begin{equation}
X_N:=\frac{U_{[N:N]}}{b_N}.
\end{equation}
We show that
\begin{equation}
\mathbb{E}[X_N]\longrightarrow1,
\qquad
\sup_{N\geq N_0}\mathbb{E}[X_N^2]<\infty
\end{equation}
for some $N_0$.

Suppose first that $F$ has a finite positive upper endpoint
\begin{equation}
x_F:=\sup\{x:F(x)<1\}\in(0,\infty).
\end{equation}
Then $b_N\to x_F$. Moreover, for every $\delta\in(0,x_F)$,
\begin{equation}
\mathbb{P}\bigl(U_{[N:N]}\leq x_F-\delta\bigr)
=
F(x_F-\delta)^N
\longrightarrow0,
\end{equation}
while $U_{[N:N]}\leq x_F$ almost surely. Hence
\begin{equation}
U_{[N:N]}\xrightarrow{\mathrm{P}}x_F
\qquad\text{and therefore}\qquad
X_N\xrightarrow{\mathrm{P}}1.
\end{equation}
Since $b_N\to x_F>0$, for all sufficiently large $N$,
\begin{equation}
0\leq X_N\leq\frac{x_F}{b_N}\leq2.
\end{equation}
Consequently,
\begin{equation}
\mathbb{E}[X_N]\longrightarrow1,
\qquad
\sup_{N\geq N_0}\mathbb{E}[X_N^2]<\infty.
\end{equation}

Now suppose that $F$ has an unbounded upper endpoint and
$V_{F}\in\operatorname{RV}_{\theta}$. Let
$E_1,\ldots,E_N$ be independent standard exponential random
variables and define
\begin{equation}
t_N:=\max_{1\leq i\leq N}E_i.
\end{equation}
The quantile representation gives
\begin{equation}
U_{[N:N]}\overset{d}{=}V_{F}(t_N),
\qquad
b_N=V_{F}(\ln N),
\end{equation}
and hence
\begin{equation}
X_N\overset{d}{=}\frac{V_{F}(t_N)}{V_{F}(\ln N)}.
\end{equation}

For every $\delta\in(0,1)$,
\begin{equation}
\mathbb{P}\left(t_N\leq(1-\delta)\ln N\right)
\leq e^{-N^\delta},
\end{equation}
and
\begin{equation}
\mathbb{P}\left(t_N>(1+\delta)\ln N\right)
\leq N^{-\delta}.
\end{equation}
Thus,
\begin{equation}
\frac{t_N}{\ln N}\xrightarrow{\mathrm{P}}1.
\end{equation}
The uniform convergence theorem for regularly varying functions
therefore yields
\begin{equation}
X_N
\overset{d}{=}
\frac{V_{F}(t_N)}{V_{F}(\ln N)}
\xrightarrow{\mathrm{P}}1.
\end{equation}

It remains to control the second moments. Fix $\eta>0$. By
Potter's bound and the monotonicity of $V_{F}$, there exist
constants $C>0$ and $N_0$ such that, for $N\geq N_0$,
\begin{equation}
X_N^2
\leq
C\left[
1+
\left(\frac{t_N}{\ln N}\right)^{2(\theta+\eta)}
\right].
\end{equation}
(when $t_N\le \ln N $ this is trivial).
For every $s>0$,
\begin{equation}
\sup_{N\geq2}
\mathbb{E}\left[
\left(\frac{t_N}{\ln N}\right)^s
\right]
<\infty.
\end{equation}
Indeed, for $t\geq2$,
\begin{equation}
\mathbb{P}\left(\frac{t_N}{\ln N}>t\right)
\leq Ne^{-t\ln N}
=N^{1-t}
\leq2^{1-t},
\end{equation}
and the claim follows from the tail-integral formula. Therefore,
\begin{equation}
\sup_{N\geq N_0}\mathbb{E}[X_N^2]<\infty.
\end{equation}
Thus $\{X_N\}$ is uniformly integrable, and since
$X_N\xrightarrow{\mathrm{P}}1$,
\begin{equation}
\mathbb{E}[X_N]\longrightarrow1.
\end{equation}

Finally, define
\begin{equation}
X_{N,j}:=\frac{U_{[N:N]}^{(j)}}{b_N},
\qquad j=1,\ldots,LN.
\end{equation}
These variables are i.i.d. copies of $X_N$. Therefore,
\begin{align*}
\mathbb{E}\left[
\left(
\frac1{LN}\sum_{j=1}^{LN}X_{N,j}-1
\right)^2
\right]
&=
\frac{\operatorname{Var}(X_N)}{LN}
+
\left(\mathbb{E}[X_N]-1\right)^2  \\
&\leq
\frac{\mathbb{E}[X_N^2]}{LN}
+
\left(\mathbb{E}[X_N]-1\right)^2
\longrightarrow0.
\end{align*}
Hence
\begin{equation}
\frac{1}{LNb_N}
\sum_{j=1}^{LN}U_{[N:N]}^{(j)}
\xrightarrow{L^2}1.
\end{equation}
\end{proof}
Now, using Lemma \ref{lemma:averaging-independent-maxima} and Claim \ref{Remark: sum-rate} we can finally prove the theorem.

Define
\begin{equation}
b_N:=F^{-1}\left(1-\frac{1}{N}\right).
\end{equation}
Since the max-min value is bounded above by the average optimal sum welfare value,
\begin{equation}
M_{N,LN}\leq \frac{1}{N}S_{N,LN},
\end{equation}
we obtain
\begin{equation}
\left\{
M_{N,LN}>\frac{Lb_N}{1-\gre}
\right\}
\subseteq
\left\{
\frac{1}{N}S_{N,LN}>\frac{Lb_N}{1-\gre}
\right\}.
\end{equation}
Therefore,
\begin{equation}
\mathbb{P}\left(
1-\frac{Lb_N}{M_{N,LN}}>\gre
\right)
\leq
\mathbb{P}\left(
\frac{1}{N}S_{N,LN}>\frac{Lb_N}{1-\gre}
\right).
\end{equation}
By Claim~\ref{Remark: sum-rate},
\begin{equation}
S_{N,LN}\stackrel{d}{=}
\sum_{j=1}^{LN}U_{[N:N]}^{(j)},
\end{equation}
where
$U_{[N:N]}^{(1)},\ldots,U_{[N:N]}^{(LN)}$
are i.i.d. copies of the maximum of $N$ independent samples from
$F$. Hence
\begin{equation}
\frac{1}{N}S_{N,LN}
\stackrel{d}{=}
\frac{1}{N}\sum_{j=1}^{LN}U_{[N:N]}^{(j)}.
\end{equation}

By Lemma~\ref{lemma:averaging-independent-maxima},
\begin{equation}
\frac{1}{LNb_N}\sum_{j=1}^{LN}U_{[N:N]}^{(j)}
\xrightarrow{\mathrm{P}}1.
\end{equation}
Equivalently,
\begin{align}
\label{eq:S_N-quantile}
    \frac{S_{N,LN}}{LNb_N}\xrightarrow{\mathrm{P}}1.
\end{align}

Since for $\gre <1$,
\begin{equation}
\frac{1}{1-\gre}>1,
\end{equation}
it follows that
\begin{equation}
\mathbb{P}\left(
\frac{1}{N}S_{N,LN}>\frac{Lb_N}{1-\gre}
\right)
=
\mathbb{P}\left(
\frac{S_{N,LN}}{LNb_N}>\frac{1}{1-\gre}
\right)
\longrightarrow 0.
\end{equation}
Therefore,
\begin{align}
    \label{eq:M_N-quantile}
\mathbb{P}\left(
1-\frac{LF^{-1}\left(1-\frac{1}{N}\right)}
     {M_{N,LN}}
>\gre
\right)
\longrightarrow 0.
\end{align}
The case $\gre\ge1$ is trivial.
This completes the proof.
\end{proof}
 \end{appendices}

\end{document}